\documentclass[11pt,a4paper, leqno]{article}                     

\usepackage[T1]{fontenc}
\usepackage{paralist}
\usepackage{graphicx}
\usepackage[english]{babel}
\usepackage{amsmath, amsopn, amstext, amscd, amsfonts, amssymb,
  amsbsy, mathrsfs}
\usepackage{enumerate}
\usepackage{url}
\usepackage[numbers]{natbib}
\usepackage{varioref}
\usepackage{dsfont}
\usepackage{empheq}
\usepackage{tikz}
\usetikzlibrary{calc,patterns}
\usepackage{float}

\newcommand{\bee}{\begin{enumerate}}
\newcommand{\eee}{\end{enumerate}}

\newcommand{\I}{{\mathcal{I}}}

\newcommand{\D}{\mathcal{D}}
\newcommand{\beqn}{\begin{equation}}
\newcommand{\eeqn}{\end{equation}}

\newcommand{\PE}[1]{{\lfloor#1\rfloor}}

\let\epsilon=\varepsilon

\let\phi=\varphi

\newcommand{ \ms }[2]{\langle #1,#2 \rangle}
\newcommand{ \abs }[1]{ \vert #1 \vert }

\newcommand{ \Zr }{\mathbb{Z}}

\newcommand{\st }{\xi}

\newcommand{ \Dom }{\mathcal{D}}

\newcommand{\Cs}{\mathcal{C}}

\newcommand{ \Lp }{\mathbb{L}}

\newcommand{ \R }{\mathbb{R}}
\newcommand{ \Rg }{\mathcal{R}}
\newcommand{ \Ke }{\mathcal{N}}
\newcommand{ \N }{\mathbb{N}}

\newcommand{ \Det }{\text{Det}}
\newcommand{ \Span }{\text{Span}}
\newcommand{ \Se }{\mathcal{S}}

\newcommand{ \Pb}{\mathbb{P}}
\newcommand{ \Q }{\mathbb{Q}}

\newcommand{ \qr }{\mathscr{Q}}
\newcommand{ \Exp }{ \mathbb{E} }

\newcommand{ \normL }[2]{ \Vert #1  \Vert_{#2} }

\newcommand\mycite[2][]{\cite[#1]{#2}}
\labelformat{section}{Section~#1}   
\labelformat{subsection}{Section~#1}   
\labelformat{equation}{eq.~(#1)}
\labelformat{align}{eq.~(#1)}
\labelformat{figure}{Figure~#1}
\labelformat{table}{Table~#1}

\usepackage{hyperref}
\hypersetup{colorlinks=true, linkcolor=blue, urlcolor=blue,
  citecolor=blue, breaklinks=true}
\usepackage{url}
\usepackage{breakurl}

\numberwithin{equation}{section}
\usepackage[perpage, symbol]{footmisc}

\usepackage[top=3cm, bottom=3cm, left=2cm , right=2cm]{geometry}
\usepackage{amsthm}
\newtheoremstyle{rem}{}{}{\itshape}{}{\large\scshape\bfseries}{.}{5pt}{}
\theoremstyle{rem}
\theoremstyle{rem}\newtheorem{theorem}{Theorem}[section]
\theoremstyle{rem}\newtheorem{lemma}{Lemma}[section]
\theoremstyle{rem}\newtheorem{proposition}{Proposition}[section]
\theoremstyle{rem}
\theoremstyle{rem}
\theoremstyle{rem}
\theoremstyle{rem}
\newtheoremstyle{ack}{}{}{}{}{\large\scshape\bfseries}{.}{5pt}{}
\theoremstyle{ack}

\title{Technical Report\\
Risk-neutral density recovery via spectral analysis}

\author{Jean-Baptiste Monnier\thanks{Office 5B01, LPMA, Universit\'{e} Paris 7, 175
rue du Chevaleret, 75013, Paris, France. Tel: +33157279164. Email: \texttt{j.r.monnier@gmail.com}}}

\begin{document}

\maketitle

\begin{abstract}
In this paper, we propose a new method for estimating the conditional
risk-neutral density (RND) directly from a cross-section of put option
bid-ask quotes. More precisely, we propose to view the RND recovery problem as an
inverse problem. We first show that it is possible to define
\textsl{restricted put and call operators} that admit a singular value decomposition (SVD),
which we compute explicitly. We subsequently show that this new framework
allows us to devise a simple and fast quadratic programming method to recover the smoothest
RND whose corresponding put prices lie inside the bid-ask
quotes. This method is termed the spectral recovery method (SRM). Interestingly, the SVD of the restricted put and call
operators sheds some new light on the RND recovery problem. The SRM improves on other RND recovery methods
in the sense that 
\begin{inparaenum}[1)]
\item it is fast and simple to implement since it requires solution of
  a single quadratic program, while being fully nonparametric;
\item it takes the bid ask quotes as sole input and does not require
  any sort of calibration, smoothing or preprocessing of the data;
\item it is robust to the paucity of price quotes;
\item it returns the smoothest density giving rise to prices
  that lie inside the bid ask quotes. The estimated RND is therefore
  as well-behaved as can be;
\item it returns a closed form estimate of the RND on the interval
  $[0,B]$ of the positive real line, where $B$ is a positive constant
  that can be chosen arbitrarily. We thus obtain both the middle part of the RND together with
  its full left tail and part of its right tail.
\end{inparaenum}
We confront this method to both real and simulated data and observe
that it fares well in practice. The SRM is thus found to be a
promising alternative to other RND recovery methods.
\end{abstract}

\textsc{key words:} 
Risk-neutral density; Nonparametric estimation;
Singular value decomposition; Spectral analysis; Quadratic programming.\\



\section{Introduction}

\subsection{The setting}
Over the last four decades, the no-arbitrage assumption has proved to
be a fruitful starting point that paved the way for the elaboration of
a rich theoretical framework for derivatives pricing known today as
\textsl{arbitrage pricing theory}. Among its numerous achievements, the
arbitrage pricing theory has set forth two fundamental theorems. The
\textsl{First Fundamental Theorem of Asset Pricing} (see
\mycite[p.72]{Musiela2008}) proves that a
market is arbitrage-free if and only if there exists a measure $\Q$
equivalent to the historical (or statistical) measure $\Pb$, which
turns the underlying price process into a martingale. $\Q$ is therefore
referred to as a \textsl{martingale measure}. The \textsl{Second
  Fundamental Theorem of Asset Pricing} (see
\mycite[p.73]{Musiela2008}) proves in turn that this martingale
measure is unique if and only if the market is \textsl{complete} (see
\mycite[p.300]{Musiela2008} for terminology). Let us
denote by $S_{\tau}$ the positive valued price of the underlying at a
deterministic future date $\tau$ and by $\pi(S_\tau)$ the payoff of a
contingent claim maturing at time $\tau$. Let us
moreover denote by $q$ the marginal density of $S_{\tau}$ under $\Q$ with respect to the
Lebesgue measure on the positive real line, assuming that it exists. As initially proved in
\mycite{Cox1976}, the arbitrage price of
this derivative security writes as its discounted expected payoff under $\Q$, that is,
\begin{align*}
e^{-r\tau}\Exp_{\Q} \pi(S_\tau) = e^{-r\tau}\int_{x \geq 0} \pi(x)
\Q(S_{\tau}\in dx) = e^{-r\tau}\int_{x \geq 0}
\pi(x) q(x) dx,
\end{align*}
where $r$ stands for the continuously compounded risk-free rate. It is
a widely acknowledged fact that financial markets are
incomplete, if only due to the presence of jumps in the
underlying price process. In such a setting, and as described above,
there exist possibly very many $q$s, and therefore, very many
corresponding systems of arbitrage-free prices. Let us denote by $\qr$
the corresponding set of valid densities $q$. The elements $q$ of
$\qr$ are most often referred
to as risk-neutral densities (RNDs) and we will stick to this
terminology in the sequel.\\
 
RNDs are of crucial interest for Central Banks and, in fact, most institutions
and people concerned with financial markets
since they represent the market sentiment about a given
underlying price process at a future point in time (see
\mycite{Bahra1997}). They are also of crucial interest to the financial
derivatives industry since the knowledge of $q$ allows to price
new derivative securities in an arbitrage-free way with respect to
traded ones. For these reasons, the literature related to
risk-neutral density estimation is very extensive, the bulk of it
dating back to the late 90's and early 2k's. It is not our purpose
here to present an exhaustive review of this literature. Excellent and
up-to-date reviews can in fact be found in \mycite{Figlewski2008,
  Jackwerth2004}. Older but still relevant ones can be found in \mycite{Cont1997, Bahra1997}.\\

Among derivative securities, call and put options play a very
particular role since they are actively traded in the market and thus believed to be
efficiently priced. Let us recall that a call of strike $\st$ and maturity
$\tau$ gives its holder the right to buy the underlying security at
maturity time $\tau$ at price $\st$. It is an insurance against a rise
in the price of the underlying. Its payoff writes $\pi(S_{\tau},\st) =
\theta(S_{\tau},\st) = (S_{\tau} - \st)^{+}$, where we have written $(x)^+ =
\max(x,0)$ for $x\in \R$. Conversely, a put option gives
the right to sell the underlying security. It is an insurance against
a fall in the underlying price and its payoff writes $\theta^*(S_{\tau}, \st)
= (\st - S_{\tau})^+$. Here and in what follows, we denote the strike
price by $\xi$ and not by $k$, which will stand for a
running index in $\N$.\\ 

According to the celebrated Breeden-Litzenberger
formula, the second derivative of put and
call prices with respect to their strike price both equal the discounted RND
$e^{-r\tau}q$ (see \mycite{Breeden1978}). Therefore, if a continuum of put or call prices were
available in the market, we would have direct access to the RND by the
latter formula. However, this is not the case and only a few strike
prices around the forward price are quoted and actively
traded at each maturity date. Depending on the market, we overall reckon
from $5$ to $50$ quotes at a given maturity date $\tau$. To complicate the
matter even more, quotes do not appear as a single price. Dealers
quote in fact a bid price, at which they offer to buy the security,
and an ask price, at which they offer to sell the security. The
difference between both prices is referred to as the bid-ask
spread. For an interesting insight into the nature of
option quotes and sources of error in them, the reader is referred to, say,
\mycite[p.786]{Hentschel2003}.

\subsection{The problem and brief literature review}
As detailed above, if traded puts and calls at a given maturity $\tau$ are arbitrage free, they
must write as their expected discounted
payoff with respect to a single RND $q$ drawn from the set $\qr$. Given the paucity of quoted option prices at a given
maturity $\tau$ and the presence of a
bid-ask spread, it is clear that many RNDs could in fact be hidden
behind quoted option prices. Therefore, the RND quest is not that much about
estimating the true RND that is used by the market for pricing
purpose, since the nature of the quotes does not allow to identify it
uniquely. It is rather more about recovering a valid RND, meaning an
actual density function, to be chosen according to a criterion
typically related to its smoothness or 
information content. Historically, three main routes have been used to
recover a RND from quoted option prices: parametric methods,
nonparametric methods and models of the underlying price
process. Each of them have their pros and cons. Parametric methods are
well adapted to small data sets and always recover a
density. However, they constrain the RND to belong to a given
parametric family. On the other hand, models of the underlying price process have been
the first great success of arbitrage pricing theory with the
celebrated geometric Brownian motion (see \mycite{Black1973,
  Merton1973}). However, the limitation of the log-normal
distribution is now widely acknowledged and no satisfying stochastic
process has yet been proposed that both reproduce accurately the dynamics of
the underlying price process and be analytically
tractable. Nonparametric methods circumvent both of these problems in the
sense that they do not require any stringent assumption on the process
generating the data (they are model-free) and can recover all
possible densities. As a main drawback, these methods are often data intensive.\\

Let us briefly come back on some contributions to
the nonparametric literature which are relevant to the present
paper. We can classify nonparametric methods as follows.
\begin{compactitem}
\item \textsl{The expansion methods}. It includes the Edgeworth (see
  \mycite{Jarrow1982}) and cumulant expansions (see
  \mycite{Potters1998}), which allow to estimate a finite number of
  RND cumulants. It also includes orthonormal basis methods such as Hermite polynomials (see
  \mycite{Abken1996}), which rely on well known Hilbert space techniques and
  give access to the middle part of the RND.  
\item \textsl{The kernel regression methods}. As a recent example, \mycite{Ait-Sahalia2003} have
  introduced a shape constrained local polynomial estimator of the
  RND. Notice that it performs estimation on the average quoted prices (that
  is, the average of the bid-ask quotes) and requires therefore to
  pre-process them in order to make them arbitrage-free. Moreover, the
  returned RND depends on the kernel chosen and it is not clear how it
  relates to the other valid RNDs in term of information content or smoothness.
\item \textsl{The maximum entropy method}. It is introduced in \mycite{Buchen1996, Stutzer1996},
  where the RND $q$ is obtained via the maximization of an entropy
  criterion. According to \mycite[p.19]{Cont1997}, this method often
  gives bumpy (multimodals) estimates
  since it imposes no smoothness restriction on the estimated
  density. In addition, it is claimed in \mycite[p.1620]{Jackwerth1996},
  that this method presents convergence issues.
\item \textsl{Other methods}, which do not belong to any of the three
  categories above. Among them, we can refer to the positive
  convolution approximation (PCA) of \mycite{Bondarenko2003}. In practice, it fits a
  finite (but large) convex linear combination of normal densities to the
  average quoted put prices and approximates the RND by the weights of
  the linear combination. It thus presents similarities with
  \mycite{Jackwerth1996}, since it ultimately fits a discrete set of
  probabilities to the average quoted prices. We can also refer
  to the smoothed implied volatility smile method (SML) as in
  \mycite{Figlewski2008}. This method uses the Black-Scholes
  formula as a non-linear transform. It consists in fitting a polynomial
  through the implied volatilities obtained from average quoted
  prices, and using the continuum of option prices obtained in that
  way to get the RND via the Breeden-Litzenberger
  formula. \mycite{Figlewski2008} refines this method by taking the
  bid-ask quotes into account at the implied volatility fit stage. The
  SML method gives access to the middle part of the
  RND. \mycite{Figlewski2008} proposes in addition a method for
  appending generalized extreme value
  (GEV) tail distributions to it. The SML method is cumbersome and can
  seem a bit odd since it requires going from price space to implied
  volatility space, back and forth. It is claimed that it is
  outperformed in term of accuracy and stability by simpler parametric methods in \mycite{Bu2007}.
\end{compactitem}

\subsection{Our results}

In this paper, we propose to view the RND recovery problem as an
inverse problem. We first show that it is possible to define
\textsl{restricted put and call operators} that admit a singular value decomposition (SVD),
which we compute explicitly. We subsequently show that this new framework
allows to devise a simple and fast quadratic programming method to recover the smoothest
RND that is consistent with market bid-ask quotes.\\

To be more precise, let us denote by $\I$ the segment $[0,B]$ of the
positive real line. We define the restricted put and call operators,
denoted by $\gamma^*$ and $\gamma$, from
$\Lp_2\I$ into itself (see \ref{eq:rescallop} and \ref{eq:resputop}
below) and show that they are conjugates of one another. We prove that
the resulting self-adjoint operator $\gamma^*\gamma$ is compact. As a
consequence of the spectral theorem (see
\mycite{Halmos1963}), $\gamma^*$
admits a singular value decomposition with positive decreasing
singular values. We prove that the corresponding
singular bases are complete in $\Lp_2\I$ (see
Theorem \ref{theorem:SVDcompact}, item \ref{item:SVDcomplete})) and compute them
explicitly together with their singular values (see
\ref{figure:eigenplots}). 
\begin{figure}[!h]
\includegraphics[width=\textwidth, height =0.45\textheight]{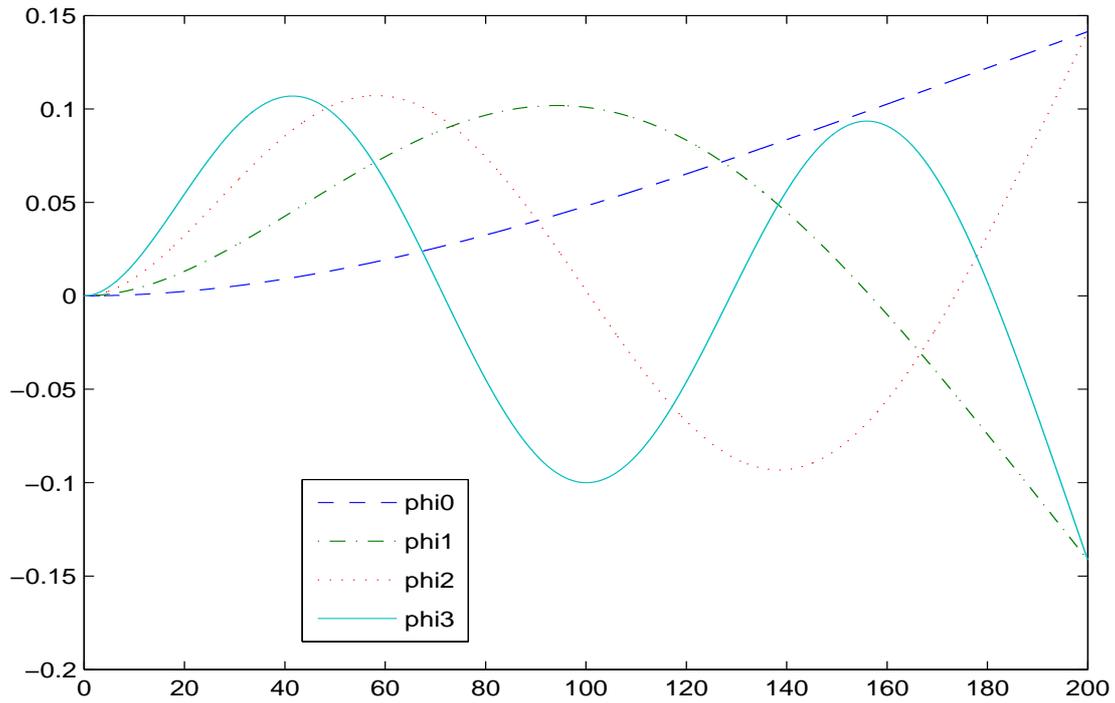}
\includegraphics[width=\textwidth, height =0.45\textheight]{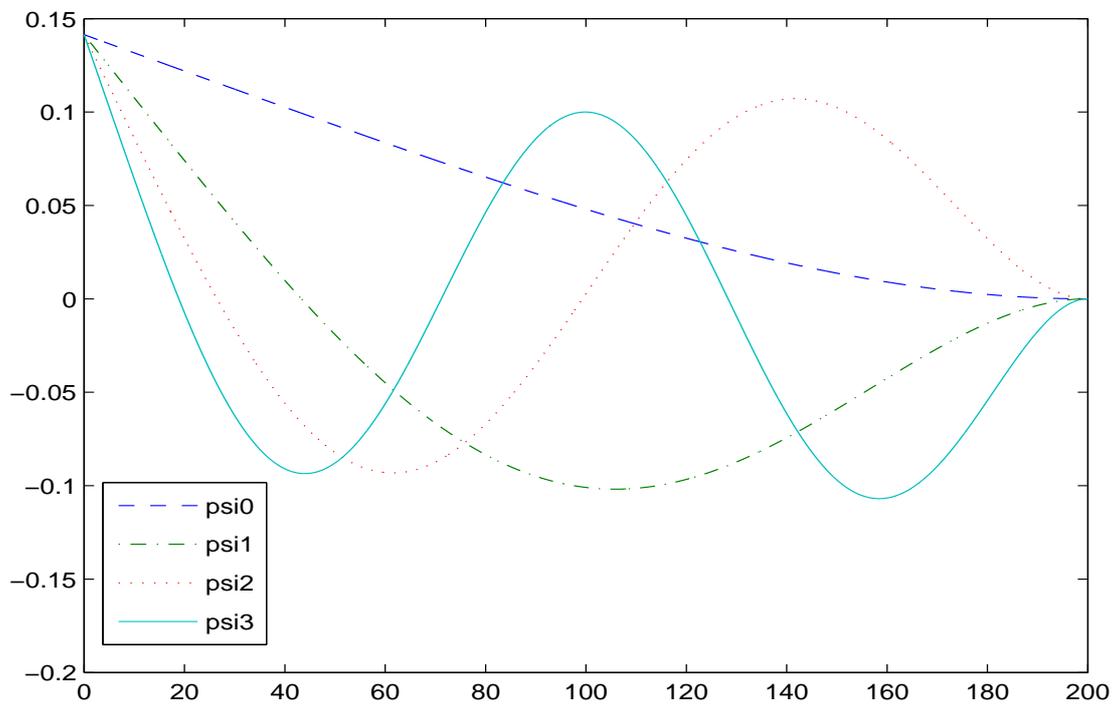}
\caption{Here, we plot the first four elements of both singular
  bases. At the top we plot $\phi_k$, $k=0,\ldots,3$. At the bottom,
  we plot $\psi_k$, $k=0,\ldots,3$.}\label{figure:eigenplots}
\end{figure} 
To fix notations, we
will write $(\phi_k)_{k\geq 0}$ and $(\psi_k)_{k\geq 0}$ the two
orthonormal families of $\Lp_2\I$ such that $\gamma^*\gamma\phi_k = \lambda_k^2
\phi_k$, $\gamma\gamma^*\psi_k = \lambda_k^2 \psi_k$, where $(\lambda_k)_{k\geq 0}$ is a positive
decreasing sequence of singular values. Precisely, we obtain explicitly,
\begin{align*}
\lambda_k = \left( \frac B{\rho_k} \right)^2,
\end{align*}
where 
\begin{align*}
\rho_k &= \frac{\pi}2 + k\pi + (-1)^{k} \beta_{k}, &k\in \N,
\end{align*}
and, for all $k \in \N$, $\beta_k$ is the smallest positive solution of the following
fixed point equation in $u$,
\begin{align*}
\exp( \pi/2 + k\pi + (-1)^k u) = \frac{1 + \cos(u)}{\sin(u)}.
\end{align*}
Interestingly, the positive sequence $(\beta_k)$ decreases
exponentially fast toward zero as detailed in Lemma
\ref{lemma:thetak}. Therefore, the sequence of singular values $(\lambda_k)_{k\geq 0}$ tends asymptotically
toward zero at a rate of order $k^{-2}$. The RND recovery problem is
therefore said to be mildly ill-posed with a degree of ill-posedness
equal to $2$ (see \mycite[p.40]{Engl1996}). Furthermore, for all $\xi \in \I$, we obtain,
\begin{align*}
\phi_k(\xi) &= \Bigl( a_{k,1}e^{\rho_k \st/B}+ a_{k,2}e^{-\rho_k \st/B}\Bigr) +
\Bigl( a_{k,3}\cos(\rho_k t/B) + a_{k,4}\sin(\rho_k \st/B)\Bigr),\\
\psi_k(\xi) &= \Bigl(a_{k,1}e^{\rho_k \st/B}+ a_{k,2}e^{-\rho_k
    \st/B}\Bigr) - \Bigl(a_{k,3}\cos(\rho_k t/B) + a_{k,4}\sin(\rho_k \st/B)\Bigr).
\end{align*}
where the coefficients $a_{k,i},i=1,\ldots,4$ are such that,
\begin{align*}
a_{k,1} &= \frac{1}{\sqrt{B}} \frac{(-1)^{k}}{e^{\rho_k} + (-1)^{k}},\displaybreak[2]\\ 
a_{k,2} &= (-1)^{k}e^{\rho_k}a_{k,1} = \frac{1}{\sqrt{B}}  \frac{1}{1 +(-1)^{k}
  e^{-\rho_k}}, \displaybreak[2]\\
a_{k,3} &= -\frac{1}{\sqrt{B}}, \displaybreak[2]\\
a_{k,4} &= \frac{1}{\sqrt{B}} \frac{1-(-1)^{k}e^{-\rho_k}}{1+(-1)^{k}e^{-\rho_k}}.
\end{align*}

Based on this new framework, we propose a spectral approach to RND
recovery. It is fully nonparametric and can recover the
restriction of any density to the interval $\I$. To that end, we
notice that the singular bases functions $\phi_k$ and $\psi_k$ are in
fact oscillations $h_{k,2}$ at frequency $\rho_k/B$ carried by the exponential
trend $h_{k,1}$ (see \ref{eq:hkdef} and \ref{eq:defrhok} for
notations). Conveniently, smooth densities
are therefore essentially captured by low singular spaces. The idea of
recovering the smoothest density among the valid ones was initially
suggested in \mycite{Jackwerth1996}. Subsequently, \mycite{Cont1997}
correctly pointed out that the smoothness criterion can be debated as
it is difficult to give it an economic or even information theoretic
meaning. Our spectral
approach sheds some new light on this issue and makes it clear that
the smoothness criterion is justified by the fact
that the restricted call and put operators behave as low-pass frequency
filters. It is therefore
illusory to look for high frequency information about the
RND in a set of quoted options prices, since this information has been drastically
attenuated by the operator. The smoothness criterion
arises therefore as a by-product of the spectral nature of the
restricted put and call operators and might well not be an intrinsic
property of the true RND. Interestingly, smooth densities are also
easier to recover by nonparametric means.\\

In what follows, we exploit the rich framework offered by
the SVD of the restricted put and call operators to recover the
smoothest RND that is compatible with market quotes. As detailed in
\ref{eq:puttocoeffs} below, the discounted restricted put
operator coincides with the put price function (as a function of the
strike) on $\I$. We therefore propose to recover the smoothest RND
such that its image by the discounted restricted put operator $e^{-r\tau}\gamma^*$ lies
in-between the bid-ask quotes (see \ref{eq:puttocoeffs}). Conveniently, the singular bases present
the property of being image of one another by second derivation
modulo a multiplication by the corresponding singular value of
$\gamma^*$ (see Theorem \ref{theorem:SVDbasisprop}). This
allows us to characterize the smoothness of the estimated RND directly in
term of a quadratic form of the coefficients of the estimated put
price function, which depends on the singular values of the restricted put
operator (see Proposition \ref{proposition:smoothness}). This crucial feature allows
to recover the smoothest RND as the solution of a simple quadratic program,
which takes the bid ask quotes as sole input. Our estimation method
improves on existing ones in several ways, which we sum up here.
\begin{compactitem}
\item It is fast and simple to implement since it only requires
  solution of a single quadratic program, while being fully nonparametric.
\item It is robust to the paucity of price quotes since the smaller
  the number of quotes, the less constrained the quadratic program and
  thus the easier to solve.
\item It takes the bid ask quotes as sole input and does not require
  any sort of smoothing or preprocessing of the data.
\item It returns the smoothest density giving rise to price quotes
  that lie inside the bid ask quotes. The estimated RND is therefore
  as well-behaved as can be.
\item It returns a closed form estimate of the RND on
  $\I$. We thus obtain both the middle part of the RND together with
  its left tail and part of its right tail. Interestingly, the left tail contains crucial
  information about market sentiments relative to a potential
  forthcoming market crash.
\end{compactitem}

It is noteworthy that the singular vectors $\phi_0$ and $\psi_0$
corresponding to the largest singular value $\lambda_0$ of $\gamma$
and $\gamma^*$ look themselves very much like cross sections of put
and call prices, respectively (see \ref{figure:eigenplots}). In that sense, they will be able to
capture the bulk of the shape of a cross section of option prices,
while the subsequent singular vectors will add corrections to this general
behavior. This is a crucial feature of this SVD that leads us
to think that the singular bases of the restricted pricing operators are
appropriate tools to recover the RND $q$. Interestingly, the
performance of our quadratic programming algorithm on real data
is indeed quite convincing (see \ref{section:simuls} for details).\\

Readers interested in appending a full right tail to this estimated RND are
referred to \mycite{Figlewski2008}, who proposes a simple method for
smooth pasting of parametric GEV tail distributions to an estimated
RND.\\



Here is the paper layout. We introduce the restricted call and put
operators, $\gamma$ and $\gamma^*$, and operators derived therefrom in
\ref{section:def}. We detail the properties of operators
$\gamma^*\gamma$ and $\gamma\gamma^*$ on the one hand, and $\gamma$ and
$\gamma^*$ on the other hand, in \ref{section:prelim} and
\ref{section:restrictedprop}, respectively. Other results relative to
these four operators are reported in
\ref{section:otherres}. \ref{section:explicitcomp} gives explicit
expressions for the $(\lambda_k)$, $(\phi_k)$ and
$(\psi_k)$. The spectral recovery method (SRM) is detailed in
\ref{section:SRM}. Finally, we run a simulation study in
\ref{section:simuls}. An
Appendix regroups some additional useful results.
 
\section{Definitions and setting}\label{section:def}

Let us define the restricted call operator on the interval $\I = [0,B]$ as the operator
$\gamma$ from $\Lp_2\I$ into $\Lp_2\I$ such that,
\begin{align}
(\gamma f)(\st) &= \int_\I \theta(\st,x)f(x)dx, &\st \in \I, f \in \Lp_2\I, \label{eq:rescallop}\\
\theta(\st,x) &= (x-\st)^+.\nonumber
\end{align}
It is a trivial fact that $\gamma f$ belongs indeed to $\Lp_2\I$. 
Let's denote by $\ms{.}{.}$ the usual scalar product on $\Lp_2\I$ and
by $\normL{.}{\Lp_2\I}$ the associated norm. Now, it is enough to
notice that for all $\st,x \in \I$, $\abs{\theta(\st,x)}\leq B$ and apply Cauchy-Schwartz inequality to obtain,
\begin{align*}
\normL{\gamma f}{\Lp_2\I}^2 \leq \int_\I d\st \left(  \int_\I dx \abs{\theta(\st,x)}
\abs{f(x)} \right)^2 \leq B^4 \normL{f}{\Lp_2\I}^2 < \infty.
\end{align*} 
The adjoint operator $\gamma^*$ of $\gamma$ is such that, for all $f,g \in \Lp_2\I$,
\begin{align*}
\ms{\gamma^* f}{g} &= \ms{f}{\gamma g} \\
&= \int_\I du  f(u) \int_\I dx \theta(u,x)g(x)\\
&= \int_\I dx g(x)\int_\I du \theta(u,x)f(u).
\end{align*}
Hence
\begin{align}
\gamma^* f(\st) &= \int_\I \theta^*(\st,x)f(x)dx, &\st\in \I, f \in \Lp_2\I, \label{eq:resputop}\\
\theta^*(\st,x) &= \theta(x,\st). \nonumber
\end{align}
So that $\gamma^*$ is nothing but the restricted put operator on the interval
$\I$. In particular, we can write
\begin{align}
\gamma^*\gamma f(\st) &= \int_\I \vartheta_1(\st,x)f(x)dx, &\st \in \I, f\in \Lp_2\I,
\label{eq:putcall}\\
\gamma\gamma^* f(\st) &= \int_\I \vartheta_2(\st,x)f(x)dx, &\st \in \I, f\in \Lp_2\I,\label{eq:callput}
\end{align}
where
\begin{align*}
\vartheta_1(\st,x) &= \int_\I du \theta^*(\st,u)\theta(u,x)\\
&= \int_\I du (\st-u)^+(x-u)^+ = \int_0^{\st\wedge x}du(\st-u)(x-u)\\
&= \st x (\st\wedge x) - (\st+x)(\st\wedge x)^2/2 + (\st \wedge x)^3/3,
\end{align*}
and
\begin{align*}
\vartheta_2(\st,x) &= \int_\I du \theta(\st,u)\theta^*(u,x)\\
&= \int_\I du (u-\st)^+(u-x)^+ = \int_{\st\vee x}^{B}du(u-\st)(u-x)\\
&= \st x(B - \st\vee x) -(\st+x)(B - \st\vee x)^2/2 + (B - \st\vee x)^3/3.
\end{align*}
Let us now turn to the detailed inspection of these operators.

\section{Results relative to $\gamma^*\gamma$ and $\gamma\gamma^*$}\label{section:prelim}
Let us denote by $\Rg(\kappa)$ the range of an operator $\kappa$ of
$\Lp_2\I$ and by $\Ke(\kappa)$ its null space (see
\mycite[p.23]{Debnath1990}).
Obviously both $\gamma^*\gamma$ and $\gamma\gamma^*$ are
self-adjoint. This translates into the fact that their kernels are
symmetric (meaning $\vartheta_i(\st,x) = \vartheta_i(x,\st)$). In addition, both $\vartheta_1$ and
$\vartheta_2$ are continuous on the bounded square $\I \times
\I$. Therefore, the associated operators are compact (see
\mycite[Ex.~4.8.4, p.172]{Debnath1990}). As such, they verify the
spectral theorem (see \mycite[Th.~4.10.1,~4.10.2, p.187-189]{Debnath1990}). 
\begin{theorem}\label{theorem:SVDcompact}
Given the operators $\gamma^*\gamma$ and $\gamma\gamma^*$ defined in
\ref{eq:putcall} and \ref{eq:callput} above, we have the following results.
\begin{compactenum}[1)]
\item The operators $\gamma^*\gamma$ and $\gamma\gamma^*$ are compact and
self-adjoint. As such, they admit countable families of
orthonormal eigenvectors $(\phi_k)$ and $(\psi_k)$ associated to the
same positive decreasing sequence of eigenvalues
$\lambda^2_k$, which are complete in $\Rg(\gamma^*\gamma)$
and $\Rg(\gamma \gamma^*)$, respectively. \label{item:spectraltheorem}
\item Besides, we have
\begin{align*}
\Rg(\gamma^*\gamma) &\subset \Lp_2\I \cap \Cs^4\I,\\
\Rg(\gamma^*\gamma) &\subset \Lp_2\I \cap \Cs^4\I,
\end{align*}
where $\Cs^4\I$ stands for the set of four times differentiable
functions on $\I$. \label{item:ranged4}
\item Furthermore, the orthonormal families $(\phi_k)$ and $(\psi_k)$
  are complete in $\Lp_2\I$. In other words, they are both orthonormal
  bases of $\Lp_2\I$. In fact, we can write 
\begin{align*}
\Lp_2\I &= \overline{\Rg(\gamma^*\gamma)} = \Span\{\phi_k, k \in \N\},\\
        &=\overline{\Rg(\gamma\gamma^*)} = \Span\{\psi_k, k \in \N \},
\end{align*}
where $\overline{\Rg(\gamma^*\gamma)}$ stands for the closure of
$\Rg(\gamma^*\gamma)$ in $\Lp_2\I$ (see \mycite[p.16]{Debnath1990})
and $\Span\{\phi_k, k\in \N\}$ for the set of (potentially infinite)
linear combinations of elements $\phi_k$. \label{item:SVDcomplete}
\item Therefore, $\gamma^*\gamma$ and $\gamma\gamma^*$ are both
  invertible and admit the fourth order differential operator
  $\partial_{\st}^4$ as an inverse (see \mycite[p.155]{Debnath1990} for
  terminology). More precisely, we have got
\begin{align*}
\partial_{\st}^4 \gamma^*\gamma f &= f, &\forall f \in \Lp_2\I,\\
\gamma^*\gamma \partial_{\st}^4f &= f, &\forall f \in \Rg(\gamma^*\gamma),
\end{align*} 
and idem for $\gamma\gamma^*$. \label{item:SVDinverse4}
\item Finally, we have the following spectral decompositions,
\begin{align*}
f &= \sum_{k\geq 0} \ms{f}{\phi_k}\phi_k, &f &\in \Lp_2 \I,\\
\gamma^* \gamma f &= \sum_{ k\geq 0} \lambda_k^2
\ms{f}{\phi_k} \phi_k, & f&\in \Lp_2\I,
\end{align*}
and
\begin{align*}
f &= \sum_{k\geq 0} \ms{f}{\psi_k}\psi_k, &f &\in \Lp_2 \I,\\
\gamma\gamma^* f &= \sum_{k\geq 0} \lambda_k^2
\ms{f}{\psi_k} \psi_k, & f&\in \Lp_2\I.
\end{align*}
\label{item:decomposition-compact}
\end{compactenum}
\end{theorem}
\begin{proof}
As detailed above, \ref{item:spectraltheorem}) follows directly from
the spectral theorem. \ref{item:ranged4}) follows directly from the
kernel representations in \ref{eq:putcall} and
\ref{eq:callput}. It can also be seen from the fact that, for any $f
\in \Lp_2 \I$, both $\gamma
f$ and $\gamma^*f$ are twice differentiable, which follows by simple
inspection of \ref{eq:rescallop} and
\ref{eq:resputop}. \ref{item:SVDcomplete}) follows directly from
Proposition \ref{proposition:complete} below. \ref{item:SVDinverse4}) is a
direct consequence of Lemma \ref{lemma:partial4} below. Finally,
\ref{item:decomposition-compact}) follows directly from
\ref{item:spectraltheorem}) and \ref{item:SVDcomplete}).   
\end{proof}

\section{Results relative to $\gamma$ and $\gamma^*$}\label{section:restrictedprop}
The following theorem details the properties of the restricted put and
call operators. It builds upon Theorem \ref{theorem:SVDcompact} above.
\begin{theorem}\label{theorem:SVD}
Given operators $\gamma$ and $\gamma^*$ defined in \ref{eq:rescallop} and
\ref{eq:resputop} above, we have the following results.
\begin{compactenum}[1)]
\item Consider the sequence of positive decreasing singular values $\lambda_k$ and
  singular vectors $(\phi_k)$ and $(\psi_k)$ defined in Theorem
  \ref{theorem:SVDcompact} above. The restricted put and call
  operators $\gamma^*$ and $\gamma$ are such that, for all $k\geq 0$,
\begin{align*}
\gamma \phi_k &= \lambda_k \psi_k, &\gamma^* \psi_k &= \lambda_k \phi_k.
\end{align*}\label{item:spectraltheorem2}
\item Besides, we have
\begin{align*}
\Rg(\gamma^*) &\subset \Lp_2\I \cap \Cs^2\I,\\
\Rg(\gamma) &\subset \Lp_2\I \cap \Cs^2\I,
\end{align*}
where $\Cs^2\I$ stands for the set of two times differentiable
functions on $\I$.\label{item:ranged2}
\item In addition, we have $\Lp_2\I = \overline{\Rg(\gamma^*)} =
  \overline{\Rg(\gamma)}$. So that both $\gamma$ and $\gamma^*$ are
  invertible and admit the second order partial differential operator
  $\partial_{\st}^2$ as an inverse. In particular, we obtain 
\begin{align}
\partial_{\st}^2 \gamma f(\st) &= \partial_{\st}^2\gamma^*f(\st) = f(\st),
&\forall f\in\Lp_2\I.\label{eq:d2}
\end{align}
So that the knowledge of $\gamma f$ or/and $\gamma^* f$ allows to
recover $f$ directly as their second derivative. This is nothing but the so-called Breeden-Litzenberger formula restricted to the
interval $\I$.\label{item:SVDinverse2}
\item We have furthermore the following spectral decompositions,
\begin{align*}
f &= \sum_{k\geq 0} \ms{f}{\phi_k}\phi_k, &f &\in \Lp_2 \I,\\
\gamma f &= \sum_{k\geq 0} \lambda_k
\ms{f}{\phi_k} \psi_k, & f&\in \Lp_2\I,
\end{align*}
and
\begin{align*}
f &= \sum_{k\geq 0} \ms{f}{\psi_k}\psi_k, &f &\in \Lp_2 \I,\\
\gamma^* f &= \sum_{k\geq 0} \lambda_k
\ms{f}{\psi_k} \phi_k, & f&\in \Lp_2\I.
\end{align*}
\label{item:decomposition}
\item Finally, we have a put-call parity on the interval that can be written as follows
\begin{align*}
(\gamma - \gamma^*) f(\st)  &= \bar m_1(f) - \st \bar m_0(f),
\end{align*}
where we have defined $\bar m_k(f) := \int_\I x^k f(x)dx$. \label{item:PCparityI}
\end{compactenum}
\end{theorem}

\begin{proof}
The proof of \ref{item:spectraltheorem2}) follows directly from
\mycite[p.37]{Engl1996}. \ref{item:ranged2}) follows by simple
inspection of \ref{eq:resputop} and
\ref{eq:rescallop}. The first part of \ref{item:SVDinverse2}) follows from the facts that
$\Rg(\gamma) = \Rg(\gamma\gamma^*)$ and $\Rg(\gamma^*) =
\Rg(\gamma^*\gamma)$ (see \ref{item:spectraltheorem2}) above) and Theorem
\ref{theorem:SVDcompact}, item \ref{item:SVDcomplete}). The second part of
\ref{item:SVDinverse2}) follows partly from Lemma \ref{lemma:partial4} below
(see Appendix) and partly from the obvious fact that $f = \gamma^* \partial_{\st}^2 f$
for all $f \in \Rg(\gamma^*)$ (idem for
$\gamma$). \ref{item:decomposition}) follows directly from
\ref{item:spectraltheorem2}) and \ref{item:SVDinverse2}). Finally,
\ref{item:PCparityI}) follows immediately from the following
obvious computations,
\begin{align*}
(\gamma - \gamma^*) f(\st) &= \gamma f(\st) - \gamma^* f(\st)\\
&= \int_\I [\theta(\st,x) - \theta^*(\st,x)] f(x)dx \\
&= \int_\I (x-\st)f(x)dx\\
&= \bar m_1(f) - \st \bar m_0(f).
\end{align*}
\end{proof}

We regroup other results relative to the above operators in the
following section. 

\section{Other results relative to $\gamma^*\gamma$,
  $\gamma\gamma^*$, $\gamma^*$ and $\gamma$}\label{section:otherres}
We prove here that both orthonormal families $(\phi_k)$
and $(\psi_k)$ are complete in $\Lp_2\I$. Other interesting results are
to be found in the Appendix. Some of them are purely technical, while some
others are of more general interest.

\begin{proposition}\label{proposition:complete}
We have got,
\begin{align*}
\Lp_2\I &=\overline{\Rg(\gamma^*\gamma)} = \Span\{\phi_k, k \geq 0 \},\\
&= \overline{\Rg(\gamma \gamma^*)}= \Span\{\psi_k, k \geq 0\},
\end{align*}
where $\overline{\Rg(\gamma^*\gamma)}$ stands for the closure of
$\Rg(\gamma^*\gamma)$ in $\Lp_2\I$ (see \mycite[p.16]{Debnath1990})
and $\Span\{\phi_k, k\in \N\}$ for the set of (potentially infinite)
linear combinations of elements $\phi_k$.
\end{proposition}
\begin{proof}
We know from \mycite[\textsection 2.3.]{Engl1996} that,
\begin{align*}
\Lp_2\I &= \overline{\Rg(\gamma^*\gamma)} \oplus^{\perp} \Ke(\gamma^*\gamma),\\
&= \overline{\Rg(\gamma\gamma^*)} \oplus^{\perp} \Ke(\gamma\gamma^*).
\end{align*}
Therefore, it is enough to show that both null-spaces reduce to the zero
element. The kernel $\Ke(\gamma^*\gamma)$ of $\gamma^*\gamma$ is constituted by the
functions $f \in \Lp_2\I$ that are solutions of 
\begin{align*}
0& = \gamma^*\gamma f(\st), & \forall \st \in \I.
\end{align*}
Deriving four times with respect to $\st$ and applying Lemma
\ref{lemma:partial4} (see Appendix) leads to $f(\st)=0, \st\in \I$. So that
$\Ke(\gamma^*\gamma) = \{0\}$. Now it is enough to notice that
$\Ke(\gamma^*\gamma) = \Ke(\gamma)$. However, we know from Lemma \ref{lemma:breve} that
$f \in \Ke(\gamma)$ if and only if $\breve f \in \Ke(\gamma^*)$ (see
\ref{eq:breve} for notation). Therefore
$\Ke(\gamma\gamma^*) = \Ke(\gamma^*)=\breve \Ke(\gamma)= \breve \Ke(\gamma^*\gamma)
= \{0\}$, where by $\breve \Ke$, we mean $\{\breve f, f \in \Ke \}$.
\end{proof}

\section{Explicit computation of $(\lambda_k)$, $(\phi_k)$ and
  $(\psi_k)$}\label{section:explicitcomp}
\subsection{Main result}
In this section, we give explicit expressions for the singular bases
and singular vectors of the restricted call and put operators. The
results are gathered below in Theorem \ref{theorem:SVDbasisprop}. Let us write
\begin{align*}
f_{k,1}(\st)&= e^{\rho_k \st/B}, & f_{k,2}(\st)&= e^{-\rho_k \st/B},\\
f_{k,3}(\st) &=\cos(\rho_k t/B),& f_{k,4}(\st)&= \sin(\rho_k \st/B),
\end{align*}
where
\begin{align}\label{eq:defrhok} 
\rho_k &= \frac{\pi}2 + k\pi + (-1)^{k} \beta_{k}, &k\in \N,
\end{align}
and, for all $k \in \N$, $\beta_k$ is the smallest positive solution of the following
fixed point equation in $u$,
\begin{align*}
\exp( \pi/2 + k\pi + (-1)^k u) = \frac{1 + \cos(u)}{\sin(u)}.
\end{align*}
Interestingly, the positive sequence
$(\beta_k)$ decreases exponentially fast toward zero as detailed in
Lemma \ref{lemma:thetak}. In addition, we write,
\begin{align}\label{eq:hkdef}
h_{k,1} &=  a_{k,1}f_{k,1} + a_{k,2}f_{k,2}, &h_{k,2} &=  a_{k,3}f_{k,3} + a_{k,4}f_{k,4},
\end{align}
where the coefficients $a_{k,i},i=1,\ldots,4$ are such that,
\begin{align*}
a_{k,1} &= \frac{1}{\sqrt{B}} \frac{(-1)^{k}}{e^{\rho_k} + (-1)^{k}},\\ 
a_{k,2} &= (-1)^{k}e^{\rho_k}a_{k,1} = \frac{1}{\sqrt{B}}  \frac{1}{1 +(-1)^{k}
  e^{-\rho_k}}, \\
a_{k,3} &= -\frac{1}{\sqrt{B}}, \\
a_{k,4} &= \frac{1}{\sqrt{B}} \frac{1-(-1)^{k}e^{-\rho_k}}{1+(-1)^{k}e^{-\rho_k}}.
\end{align*}
Then, we have the following theorem.
\begin{theorem}\label{theorem:SVDbasisprop}
The eigenvectors $(\phi_k)$ of $\gamma^*\gamma$ and $(\psi_k)$ of
$\gamma \gamma^*$ are such that
\begin{align}\label{eq:phipsidef}
\phi_k &= h_{k,1} + h_{k,2}, & \psi_k &= h_{k,1} -h_{k,2}.
\end{align}
They are related by the following relationships,
\begin{align}\label{eq:phipsispectral}
\gamma \phi_k &= \lambda_k \psi_k, &\gamma^*\psi_k = \lambda_k \phi_k,
\end{align}
where we have written
\begin{align}\label{eq:eigenval}
\lambda_k = \left( \frac B{\rho_k} \right)^2,
\end{align}
and $\rho_k$ is defined in \ref{eq:defrhok}. They verify
$\normL{\phi_k}{\Lp_2\I}=\normL{\psi_k}{\Lp_2\I}=1$. Moreover, we have
\begin{align}\label{eq:phipsiboundary}
\psi_k(B) = \psi_k'(B) &= 0,  &\phi_k(0) &= \phi_k'(0) = 0,
\end{align}
together with
\begin{align}\label{eq:phipsibreve}
\breve \psi_k &= (-1)^k \phi_k, & \breve \phi_k &= (-1)^k\psi_k,
\end{align}
where we have written $\breve \psi_k(\st) = \psi_k(B-\st)$. And finally, we obtain as a direct consequence of \ref{eq:d2} above that
\begin{align*}
\lambda_k \partial_{\st}^2\psi_k &= \partial_{\st}^2 \gamma \phi_k =
\phi_k,\\
\lambda_k \partial_{\st}^2\phi_k &= \partial_{\st}^2 \gamma^* \psi_k =
\psi_k.
\end{align*}
\end{theorem}
\begin{proof}
Notice readily that \ref{eq:phipsiboundary}, \ref{eq:phipsibreve} and the fact that both
$\phi_k$ and $\psi_k$ are unit normed are straightforward consequences
of \ref{eq:phipsidef}. In addition, \ref{eq:phipsispectral} is a
repetition of Theorem \ref{theorem:SVD}, item \ref{item:spectraltheorem2}). So
that we are in fact left with proving \ref{eq:phipsidef} and \ref{eq:eigenval}. Each
eigenvector $f$ of $\gamma^*\gamma$ associated to the eigenvalue $r^4$ is solution of the problem, 
\begin{align}\label{eq:eigenvector} 
r^4 f &= \gamma^*\gamma f,
\end{align}
for some $r\neq 0$ and $f\in \Lp_2\I$. After differentiating four times the latter equation with respect to
$\st$ (assuming that $f \in \Lp_2\I \cap \Cs^4\I$) and applying Lemma
\ref{lemma:partial4}, we obtain that the solutions of
\ref{eq:eigenvector} are also solutions of the following fourth order ordinary
differential equation,
\begin{align*}
r^4 d_{\st}^4f - f =0,
\end{align*}
where $d_{\st}^4$ stands for the fourth order ordinary differential
operator. Its characteristic polynomial admits four roots $\pm r^{-1}$ and $\pm
i r^{-1}$. Consequently, the real solutions of the above ordinary differential
equation are of the form
\begin{align}\label{eq:genericsol}
f(\st) = b_1 e^{\st/r} + b_2 e^{-\st/r} + b_3 \cos(\st/r) + b_4 \sin(\st/r).
\end{align}
The $\phi_k$s are thus of this form. Plugging this generic solution
back into \ref{eq:eigenvector} leads in turn, after tedious but
straightforward computations, to 
\begin{align}\label{eq:systconst}
M b =0,
\end{align}
where $b$ is a $4\times 1$ vector such that $b^T = \begin{bmatrix} b_1
  & b_2 & b_3 &b_4 \end{bmatrix}$ and
$M$ is the $4 \times 4$ matrix defined by
\begin{equation}\label{eq:matconstraint}
M(r,B)= \begin{bmatrix}
 r^{-1}e^{B/r}&- r^{-1} e^{- B/r}& r^{-1}\sin \left(  B/r \right) &-r^{-1}\cos \left( B/r \right) \\ 
-r ^{-2}e^{B/r}&-r^{-2}e^{-B/r}&r^{-2}\cos \left(  B/r \right) &r^{-2}\sin \left( B/r \right) \\
r^{-3} &-r^{-3} & 0  & r^{-3}\\
r^{-4} & r^{-4} &r ^{-4} & 0 \\
\end{bmatrix}.
\end{equation}
There exists a non-trivial solution to \ref{eq:systconst} if and only
if $r$ is such that the determinant of $M$ cancels, that is
$\Det(r,M)=0$. As detailed in Proposition
\ref{proposition:eigenval}, the roots of $\Det(r,M)=0$ are exactly the
$r_m = B/\nu_m$ where $\nu_m$ is defined in \ref{eq:defnuk}. In
addition, we prove in Proposition \ref{proposition:kernelsol} that the system
$M(r_m,B)b=0$ admits the unique solution $b_m$. Reading off
\ref{eq:genericsol}, we obtain that the eigenvector of $\gamma^*\gamma$
associated to eigenvalue $r_m^4$ writes as $\alpha_m = \eta_{m,1} +
\eta_{m_2}$ where both $\eta_{m,1}$ and $\eta_{m,2}$ are defined in
\ref{eq:etakdef}. Now, it is enough to notice that, given the
properties of the sequence $(\nu_m)$ detailed in Proposition \ref{proposition:defnuk}, $r_{2k+1}^4
= r_{2k}^4$ and $r_{2k+2}^4 < r_{2k+1}^4$, $k\in \N$. In addition, we
know from Lemma \ref{lemma:redundantsol} that $\alpha_{2k+1} =
\alpha_{2k}$. This allows us to conclude that the eigenvalues of
$\gamma^*\gamma$ are, without redundancy, the $\lambda_k^2$, $k\in\N$, defined in
\ref{eq:eigenval} and the associated eigenspaces are unit-dimensional
and respectively spanned by the eigenvectors $\phi_k$, $k\in \N$, defined in \ref{eq:phipsidef}.\\
Computing $\psi_k = \lambda_k^{-1}\gamma\phi_k$ leads, after tedious but straightforward
computations to $\psi_k = h_{k,1} - h_{k,2}$ and concludes the proof.
\end{proof}

\subsection{Additional results}
This section contains a series of results that are used throughout the proof
of Theorem \ref{theorem:SVDbasisprop} above. In this section we
make use of the map $E: \N \mapsto \N$ such that $E(2k+1) = E(2k) = k$ for
all $k \in \N$.

\begin{proposition}\label{proposition:eigenval}
Let $M(r,B)$ be the $4\times 4$ matrix defined in
\ref{eq:matconstraint}. The set of solutions $r$ to the problem $\Det
M(r,B) =0$ is countable. Let us denote them by $r_m, m \in \N$. For any $m \in \N$, the solution $r_m$ can be written as
\begin{align*}
r_m &= \frac{B}{\nu_m},
\end{align*}
where $\nu_m$ is defined in \ref{eq:defnuk}. We obtain in fact that,
\begin{align*}
\Det M(r_m,B)&= 0 &\Leftrightarrow& &e^{\nu_m} &= -\frac{1 + (-1)^{E(m)} \sin(\nu_m)}{\cos(\nu_m)}.
\end{align*}
Besides, the following relationships hold true
\begin{align}
\cos\nu_m &:= -\frac{2}{e^{\nu_m} + e^{-\nu_m}} =
-\frac{1}{\cosh\nu_m}, \label{eq:equivtrigexp1}\\ 
\sin\nu_m &:= -(-1)^{E(m)} + (-1)^{E(m)}\frac{2}{1 + e^{-2\nu_m}}.\label{eq:equivtrigexp2}
\end{align}
\end{proposition}
\begin{proof}
It follows from straightforward computations that,
\begin{align}\label{eq:detexpression}
\Det M(r,B) &= 2 e^{-B/r} \left(\cos\left( B/r \right) 
 \left( e^{B/r} \right) ^{2} + 2 e ^{B/r}+ \cos \left( B/r \right) \right).
\end{align}
Let us write $\nu := B/r$ and notice that if $\cos(\nu) = 0$, then
$\Det M(r,B) =2 \neq 0$ so that we must have $\cos \nu \neq 0$ for
\ref{eq:systconst} to admit a non-trivial solution. To be more specific
$\Det M(r,B) = 0$ reduces to $P(e^{\nu}) = 0$ where $P(x) :=
\cos\left( \nu \right) x^{2} + 2 x + \cos \left( \nu \right)$.
However the roots of $P$ are given by
\begin{align*}
\delta_{\pm}(\nu) &:= \frac{ -1 \pm \sin(\nu) }{\cos(\nu)}.
\end{align*}
Henceforth, $r = B/\nu$ cancels $\Det M(r,B)$ if and only if $\nu$ is solution
of anyone of the two following fixed point equations,
\begin{align*}
e^{\nu} &= \frac{ -1 + \sin(\nu) }{\cos(\nu)}, &e^{\nu} &= \frac{ -1 - \sin(\nu) }{\cos(\nu)}.
\end{align*}
The proof follows now directly from Proposition \ref{proposition:defnuk}.
\end{proof}

\begin{proposition}\label{proposition:kernelsol}
For any $r_m$ solution of the equation $\Det M(r_m,B)=0$ (see Proposition
\ref{proposition:eigenval} above), the null space of
$M(r_m, B)$ is of dimension $1$ and is spanned by the vector 
\begin{align*}
b_m^T = \begin{bmatrix}
b_{m,1}&b_{m,2}&b_{m,3}&b_{m_4}\end{bmatrix},
\end{align*}
where we have written,
\begin{align*}
b_{m,1} &= \frac{1}{\sqrt{B}} \frac{(-1)^{E(m)}}{e^{\nu_m} + (-1)^{E(m)}},\\ 
b_{m,2} &= (-1)^{E(m)}e^{\nu_m}a_{m,1} = \frac{1}{\sqrt{B}}  \frac{1}{1 +(-1)^{E(m)}
  e^{-\nu_m}}, \\
b_{m,3} &= -\frac{1}{\sqrt{B}}, \\
b_{m,4} &= \frac{1}{\sqrt{B}} \frac{1-(-1)^{E(m)}e^{-\nu_m}}{1+(-1)^{E(m)}e^{-\nu_m}}, 
\end{align*}
and $\nu_m$ is defined in \ref{eq:defnuk}.
\end{proposition}
\begin{proof}
It is a matter of straightforward linear algebra and thus left to the
reader. Notice however, that it relies on the use of both
\ref{eq:equivtrigexp1} and \ref{eq:equivtrigexp2}.
\end{proof}

\begin{lemma}\label{lemma:redundantsol}
Let us write
\begin{align*}
\zeta_{m,1}(\st)&= e^{\nu_m \st/B}, & \zeta_{m,2}(\st)&= e^{-\nu_m \st/B},\\
\zeta_{m,3}(\st) &=\cos(\nu_m\st/B),& \zeta_{m,4}(\st)&= \sin(\nu_m\st/B),
\end{align*}
where $\nu_m$ is defined in \ref{eq:defnuk}. In addition, we write,
\begin{align}\label{eq:etakdef}
\eta_{m,1} &=  b_{m,1}\zeta_{m,1} + b_{m,2}\zeta_{m,2}, &\eta_{m,2} &=  b_{m,3}\zeta_{m,3} + b_{m,4}\zeta_{m,4},
\end{align}
where the coefficients $b_{m,i},i=1,\ldots,4$ are defined in Proposition
\ref{proposition:kernelsol}. For all $k \in \N$, we have the following relationships
\begin{align*}
\eta_{2k+1,1} &= \eta_{2k,1}, &\eta_{2k+1,2} &= \eta_{2k,2}
\end{align*}
\end{lemma}
\begin{proof}
It follows from straightforward computations using the fact that
$\nu_{2m+1} = - \nu_{2m}$.
\end{proof}

\begin{proposition}\label{proposition:defnuk}
Let us define the map $E:\N \mapsto \N$ such that $E(2k) = E(2k+1) =k$
for $k\in\N$. Let us write
\begin{align*}
g(\nu) &= \frac{-1 + \sin \nu}{\cos \nu}, & h(\nu)&= \frac{-1 - \sin\nu}{\cos\nu},
\end{align*}
and consider the fixed point equations $e^{\nu} = g(\nu)$ and $e^{\nu}
= h(\nu)$. The set of corresponding solutions is exhausted by the
sequence
\begin{align}\label{eq:defnuk}
\nu_m &= (-1)^m \left( \frac{\pi}2 + E(m)\pi + (-1)^{E(m)}
  \beta_{E(m)}\right), &m\in \N. 
\end{align}
where $(\beta_m)$ is defined in Lemma
\ref{lemma:thetak}. In particular, notice that $\nu_{2k+1} = -\nu_{2k}$ and
$\abs{\nu_{m_1}} < \abs{\nu_{m_2}}$ for all $m_1,m_2 \in \N$ such
that $E(m_1)<E(m_2)$. Notice in addition that, by construction, $\nu_m$
is solution of
\begin{align*}
e^{\nu_m} &= - \frac{1 + (-1)^{E(m)} \sin \nu_m}{\cos \nu_m}.
\end{align*}
This latter result, together with the fact that $\Det M(B/\nu_m, B) =
0$ (see \ref{eq:detexpression}), leads straightforwardly to the
following relationships,
\begin{align*}
\cos\nu_m &:= -\frac{2}{e^{\nu_m} + e^{-\nu_m}} =
-\frac{1}{\cosh\nu_m},\\
\sin\nu_m &:= -(-1)^{E(m)}  + (-1)^{E(m)}\frac{2}{1 + e^{-2\nu_m}}.
\end{align*}
\end{proposition}

\begin{proof}
Consider the fixed point equation $g(\nu) =
e^{\nu}$. Given the properties of $g$ detailed in Proposition \ref{proposition:variationg}, two cases arise depending
whether $\nu$ is positive or negative. In the case where $\nu$ is
positive, the exponential map meets $g$ at points of the form
$p_m = \frac{3\pi}2 + 2m\pi - u_m$ for $m \in \N = \{0, 1, 2, \ldots
\}$ and some small but positive $u_m$s. A direct application of Lemma \ref{lemma:fpsym} shows that the negative
solutions are exactly the $-p_m, m\in \N$.\\
The second fixed point equation $h(\nu) = e^{\nu}$ can be rewritten
as $g(-\nu) = e^{\nu}$. The positive solutions are of the form $q_m =
\frac{\pi}2 + 2m\pi + v_m, m\in \N$. And, from Lemma \ref{lemma:fpsym} again, the
corresponding negative solutions are the $-q_m, m\in \N$.\\
Let us write $t_m = \frac{\pi}2 + m\pi + (-1)^m \beta_m,
m\in \N$. It is clear that $t_{2k} = q_k$ and $t_{2k+1} = p_k$ for $k
\in \N$. In particular, $t_m$ is solution of
\begin{align}\label{eq:tm}
e^{t_m} &= - \frac{1 + (-1)^m \sin t_m}{\cos t_m}
\end{align}
Let us define the map $E: \N \mapsto \N$ such that $E(2k+1) = E(2k) =
k$ for all $k\in \N$. We define $\nu_m, m\in \N $ such that $\nu_{m} =
(-1)^m t_{E(m)}$, that is $\nu_{2k} = t_k$ and $\nu_{2k+1} = -t_k$, $k\in \N$. By construction, $\nu_m$ exhausts the set of
solutions of both fixed point equations $e^{\nu} = g(\nu)$ and
$e^{\nu} = h(\nu)$. In fact, $\nu_m$ is solution of 
\begin{align*}
e^{\nu_m} &= - \frac{1 + (-1)^{E(m)} \sin \nu_m}{\cos \nu_m}
\end{align*}
\end{proof}

\begin{proposition}\label{proposition:variationg}
Notice readily that $h(\nu) = g(-\nu)$, so that it is enough to
study the properties of $g$ alone. We have the
following results,
\begin{compactenum}
\item $g$ is defined on the domain $\D_g = \R \backslash
  \{\frac{3\pi}2 + 2m\pi, m\in \Zr \}$;
\item $g$ is $2\pi$ periodic and such that, for all $\nu
  \in \Se_g = (-\frac{\pi}2, \frac{3\pi}2)$, $g(\nu + 2m\pi) =
  g(\nu)$;
\item  Finally, $g$ is strictly increasing on $\Se_g$ and such that,
\begin{align*}
\lim_{\nu \to_{\oplus} -\frac{\pi}2}g(\nu) &= -\infty, & g(\frac{\pi}2) &= 0,
&\lim_{\nu \to_{\ominus} \frac{3\pi}2}g(\nu) &= +\infty.
\end{align*}
where we write $\to_{\oplus}$ (resp. $\to_{\ominus}$) to mean the limit from the above
(resp. below).
\item Notice that $\R \backslash \D_g$ (resp. $\R \backslash \D_h$) corresponds exactly to the set
  of all the zeros of $h$ (resp. $g$). Thus $\D_g \cap \D_h$ is the
  subset of $\R$ containing all the points where both $g$ and $h$ are
  well defined and different from zero.   
\end{compactenum}
\end{proposition}
\begin{proof}
Let us first focus on the domain of $g$. It is defined on $\R \backslash
  \{\frac{\pi}2 + m\pi, m\in \Zr \}$. However, $g$ can be extended by
  continuity to be worth zero at points $\tfrac{\pi}2 + 2m\pi, m
  \in\Zr$. Notice indeed that for any small positive $u$ and $\ell \in
  \N$, one has got
\begin{align*}
g(\frac{\pi}2 + (-1)^{\ell}u) &= \frac{-1 + \cos u}{-(-1)^{\ell}\sin u}\\
&= \frac{-\frac{u^2}2 + O(u^4)}{-(-1)^{\ell}u + O(u^3)} = (-1)^{\ell}\frac{u}2 + O(u^3).
\end{align*}
With a slight abuse of notations, we denote the latter
  extension by $g$. So that $g$ is actually defined on $\R \backslash
  \{\frac{3\pi}2 + 2m\pi, m\in \Zr \}$. The other properties
  are straightforward.
\end{proof}

\begin{lemma}\label{lemma:fpsym}
Recall that $\D_g$ and $\D_h$ are defined in Proposition
\ref{proposition:variationg}. Notice first that $\D_g \cap \D_h$ is symmetric, meaning that if
$\nu \in \D_g\cap \D_h$, then $-\nu \in \D_g\cap\D_h$. For any $\nu \in \D_g \cap \D_h$, we have the following results,
\begin{compactenum}
\item If $\nu$ is solution of the fixed point equation $e^{\nu} =
g(\nu)$, then $-\nu$ is also a solution.  
\item If $\nu$ is solution of the fixed point equation $e^{\nu} =
h(\nu)$, then $-\nu$ is also a solution.  
\end{compactenum}
\end{lemma}
\begin{proof}
Notice first that we have the identity $h(\nu)g(\nu)=1$ for any $\nu
\in \D_g \cap \D_h$. Its proof is
immediate. And therefore, for any $\nu \in \D_g\cap \D_h$ solution of $e^{\nu} = g(\nu)$, we obtain
$g(-\nu) = h(\nu) = g(\nu)^{-1} = e^{-\nu}$. And idem for the solutions of $e^{\nu} = h(\nu)$.
\end{proof}

\begin{lemma}\label{lemma:thetak}
The sequence $(\beta_k)$ is such that, for all $k\in \N$, $\beta_k$ is
the smallest positive solution of the following
fixed point equation in $u$,
\begin{align*}
\exp( \pi/2 + k\pi + (-1)^k u) = \frac{1 + \cos(u)}{\sin(u)}.
\end{align*}
In addition, the approximation $\beta_k \approx 2e^{-\frac{\pi}2 -
  k\pi}$ holds true with a large degree of accuracy from $k = 1$ onward.
\end{lemma}
\begin{proof}
Let us write $t_k = \frac{\pi}2 + k\pi + (-1)^ku$, for some small but
positive $u$ such that $t_k$ is solution of \ref{eq:tm}. Notice that 
\begin{align*}
\cos\left(\frac{\pi}2 + k\pi +(-1)^ku \right) &= -\sin(u) = -u + O(u^3),\\
\sin\left(\frac{\pi}2 + k\pi +(-1)^k u\right) &= (-1)^k\cos(u) = (-1)^k + O(u^2),\\
\exp\left(\frac{\pi}2 + k\pi +(-1)^ku\right) &= e^{\frac{\pi}2 + k\pi} (1 +(-1)^ku + O(u^2)).
\end{align*}
So that \ref{eq:tm} reduces to
\begin{align*}
\exp( \pi/2 + k\pi + (-1)^k u) = \frac{1 + \cos(u)}{\sin(u)}.
\end{align*}
Plugging-in the Taylor expansions above, we obtain
\begin{align*}
e^{\frac{\pi}2+k\pi}(1 +(-1)^k u + O(u^2)) = \frac{2 + O(u^2)}{u + O(u^3)} = \frac 1u
(2 + O(u^2)),
\end{align*}
which can be rewritten as
\begin{align}\label{eq:approxtheta}
u = e^{-\frac{\pi}2 -k\pi}(2 + O(u)).
\end{align}
It can be verified numerically that $2e^{-\frac{\pi}2 -k\pi}$ is a
very good approximation of $\beta_k$ as soon as $k\geq 1$ in the
sense that \ref{eq:tm} holds true with a very large degree of accuracy.   
\end{proof}

\section{The spectral recovery method (SRM)}\label{section:SRM}
In this Section, we first describe how $\gamma$ and $\gamma^*$ relate
to the bid-ask quotes. We then show that the SVD
of the restricted pricing operators described above can be used to
design a simple quadratic program that recovers the smoothest RND
compatible with market quotes.

\subsection{From $\gamma$ and $\gamma^*$ to call and put prices}\label{section:optoprice}
Let us denote by $P(\st)$ and $C(\st)$ the put and call prices at strike
$\st$ and by $q$ the corresponding risk neutral density. Let us
furthermore write $\bar \I = \R^+\backslash \I = (B,\infty)$. We assume that
the restriction $q_{\vert \I}$ to the interval $\I$ of $q$ is in
$\Lp_2\I$. For all $\st\in \I$, the following relationships are immediate.
\begin{align}
e^{r \tau}P(\st)&= \gamma ^*q(\st), \label{eq:puttocoeffs}\\
 e^{r\tau} C(\st) &= \gamma q(\st) +
\int_B^{\infty}(x-\st)q(x)dx \nonumber\\
&=\gamma q(\st) + m_1(q) - \st m_0(q), \label{eq:calltocoeffs}
\end{align}
where we have defined, 
\begin{align*}
m_k(f) &= \int_{\bar \I} x^kf(x)dx.
\end{align*}
Notice in particular that
\begin{align*}
m_0(q) &= \Q(S_{\tau} \geq B)=  1 - \bar m_0(q), \\
m_1(q) &= \Exp_\Q(S_{\tau} \vert S_{\tau}\geq B) \Q(S_{\tau}\geq B) = \Exp_{\Q} S_{\tau} -
\bar m_1(q).
\end{align*}
\Ref{eq:puttocoeffs} shows that put
prices directly relate to the restricted put
operator. From an estimation perspective, this is a crucial feature
that will allow us to recover the RND directly from market put
quotes. Unfortunately, the situation is slightly different for call
prices. As shown from \ref{eq:calltocoeffs}, call
prices relate to the restricted call operator via $m_1(q)$ and
$m_0(q)$, which are both unknown. Although, they could be estimated
and give rise to an estimator of the RND based on quoted call prices, 
we wont pursue this route here, but rather focus on the simpler
relation given by \ref{eq:puttocoeffs}.

\subsection{A refresher on no-arbitrage constraints}\label{section:refresherNA}

For a detailed review of model-free no-arbitrage constraints, the
reader is referred to \mycite[p.32,
\textsection~1.8]{Musiela2008} and \mycite{Davis2007}. Let us denote by $S_0$ the price today
of the underlying stock. Let us moreover assume that it pays a
continuous dividend yield $\delta$. Let us denote by $r$ the
continuously compounded short rate and by $\tau$ the time to
maturity. Let us recall first that, by no-arbitrage, put and call
prices are related by the put-call parity.
\begin{align}\label{eq:PCparity}
C(\st) - P(\st) &= S_0e^{-\delta\tau} - \st e^{-r\tau}.
\end{align}
Besides $C(0)= S_0$ and $P(0) = 0$. Let us now focus on put prices. We have,
\begin{align}
\max(0, \st e^{-r\tau} -S_0e^{-\delta\tau})&\leq P(\st) \leq \st e^{-r\tau}, \label{eq:nap0}\\
0 &\leq \partial_{\st} P(\st) \leq e^{-r\tau}, \label{eq:nap1}\\
0 & \leq \partial_{\st}^2 P(\st).\label{eq:nap2}
\end{align}
Assume we are given an increasing sequence of $n$ strikes
$\st_1<\st_2<...<\st_n$ and a set of corresponding put prices
$m_1,\ldots,m_n$. As described in \mycite{Ait-Sahalia2003}, the above
no-arbitrage relationships translate into a finite set of affine constraints
on the latter put prices. These constraints can in fact be written in matrix form
as $Am \leq b_p$, where $A$ stands for a $2n\times n$ matrix, $m$ is
the $n\times 1$ vector such that $m^T = \begin{bmatrix}m_1& \ldots&
  m_n \end{bmatrix}$ and $b_p$ is a $2n\times 1$ vector. More precisely, \ref{eq:nap2} translates into $n-2$
constraints as,
\begin{align*}
[Am]_{i} &:=  \frac{m_{i+1} - m_i}{\st_{i+1}-\st_i} - \frac{m_{i+2} - m_{i+1}}{\st_{i+2} -
  \st_{i+1}}\leq 0 := [b_p]_i, &i&=1,2,\ldots,n-2
\end{align*}
Moreover, the left-hand-side of \ref{eq:nap0} is fully captured in-sample by adding the
following additional $n$ constraints,
\begin{align}\label{eq:constLB}
[Am]_{i+n-2}&:= -m_i \leq - \max(0, \st_ie^{-r\tau} - S_0e^{-\delta\tau}):= [b_p]_{i+ n-2},
& i=1,\ldots, n
\end{align}
The right-hand-side of \ref{eq:nap0} need not be taken into account at this
stage. It is indeed less stringent than the upper-bound constraints we
will impose in the next section. Finally, given the
first $n-2$ constraints, \ref{eq:nap1} reduces to two additional
constraints,  
\begin{align*}
[Am]_{2n-1}&:= \frac{m_n - m_{n-1} }{\st_n - \st_{n-1}} \leq e^{-rT}:= [b_p]_{2n-1},\\
[Am]_{2n} &:= m_{1} - m_{2} \leq 0:=[b_p]_{2n}.
\end{align*}
Finally, let us recall that if the forward price $F_0$ of the
underlying stock is observable today, then, by no-arbitrage, it must
be equal to $S_0e^{(r-\delta)\tau}$.

\begin{figure}[tb]
\begin{center}
\begin{tikzpicture}
\draw[->] (5,0)--(10,0) node[right] {$\xi$};
\draw[->] (0,0)--(0,5) node[above] {$P(\xi)$};

\coordinate (a) at (5,0);
\node (b) at (a) [below] {$S_0e^{-\delta \tau}$};
\fill (a) circle (1pt);

\coordinate (c0) at (10,5);

\draw[loosely dotted, ultra thick, red, double distance = 2pt] (3,0.5)
.. controls +(10:1cm) and +(200:1cm) ..node[sloped,
above]{$y_i^{Ask}$} node[sloped, below]{$y_i^{Bid}$} (5,1) .. controls +(20:1cm)
and +(220:1cm) .. (7,2.4); 

\draw[dotted, thick] (0,0) .. controls +(0:.2cm) and +(190:1cm) ..  (3,0.5)
.. controls +(10:1cm) and +(200:1cm) .. (5,1) .. controls +(20:1cm)
and +(220:1cm) .. (7,2.4).. controls +(40:1cm) and +(225:1cm) .. (9,4.2);

\fill (9,0) circle (1pt);
\node (e) at (9,0) [below] {$\xi_n = B$};

\draw[dashed] (9,0) -- (9, 4);

\draw[ultra thick, dotted, blue] (0,0)--(5,0)--node[below, sloped]{$(\xi_i e^{-r\tau} - S_0 e^{-\delta \tau})^+$}(9,4);

\coordinate (o) at (0,0);
\fill[blue] (o) circle (1.5pt);
\node at (o) [above right, blue] {$y_1^{Ask}$};
\node at (o) [below] {$\xi_0 = 0$};

\coordinate (o2) at (9,4.4);
\fill[blue] (o2) circle (1.5pt);
\node at (o2) [above, blue] {$y_n^{Ask}$};

\end{tikzpicture}
\end{center}
\caption{This graph sums up the set of constraints verified by estimated
put prices, which are solutions of the quadratic optimization problem
described in \ref{eq:optimS1}. Estimated put prices $m_1, \ldots, m_n$
on the ``dense'' grid $\xi_1, \ldots, \xi_n$ are displayed as
black dots. They must lie in-between the bid-ask quotes, which are
represented by thick red dots ranging over quoted strikes $\xi_{i_1},
\ldots, \xi_{i_s}$, which correspond to a sparse subset of the
underlying dense grid $\xi_1, \ldots, \xi_n$. In addition, extreme put
prices $m_1$ and $m_n$ are bounded above by $y_1^{Ask}=0$ and
$y_n^{Ask}$, respectively, where the value of $y_n^{Ask}$ is given in
\ref{section:BAconst}. Both $y_1^{Ask}$ and $y_n^{Ask}$ appear as
thick blue dots at strikes $\xi_1=0$ and $\xi_n=B$, respectively. $m_1, \ldots, m_n$
must also verify the in-sample constraints described by the lhs of
\ref{eq:nap0}. In particular, the lhs of \ref{eq:nap0}
ensures that the $m_i$s  are lower-bounded by the $(\xi_ie^{-r\tau} -
S_0e^{-\delta\tau})^+$s, which appear as thick blue dots. Since this
lower-bound is worth $0$ for $i=1$, this, together with the
upper-bound $y_1^{Ask}=0$ actually impose $m_1=0$. Finally, $m_1,
\ldots, m_n$ verify both \ref{eq:nap1}
and \ref{eq:nap2} above. The latter constraint imposes in-sample convexity.}\label{figure:constsum}
\end{figure}

\subsection{Bid-ask spread constraints}\label{section:BAconst}
Let us assume that the market provides us with an increasing sequence of
strike prices $\st_1 < \st_2 < \ldots < \st_s$, where $s$ typically ranges from $5$ to $50$ depending on the
underlying. In addition, the market provides us with a corresponding
sequence of bid ask quotes for put options. Let us denote them by $y_1^{Ask}, \ldots,
y_s^{Ask}$ and $y_1^{Bid}, \ldots, y_s^{Bid}$. We want the
corresponding fitted put prices $(m_i)$ to lie inside the bid ask
quotes. This corresponds to the following $2s$ affine constraints,
\begin{align}\label{eq:bidaskconst}
m_i &\leq  y_i^{Ask}, &-m_i &\leq -y_i^{Bid}, &i=1, \ldots s.
\end{align}
The quoted strikes might possibly span a very small portion of the
segment $\I$ on which we want to recover the RND. In order to improve
the quality of our estimator, we can constrain it to verify the above no-arbitrage constraints on a
denser set of strikes than the quoted ones. Let us denote by $\st_1 <
\st_2 < \ldots< \st_n$ this new set of strike prices, such that
$\st_1=0$, $\st_n=B$ and including the initial quoted strikes. For
later reference, we denote by $I=\{i_1, \ldots, i_s \}$
the subset of $\{1,\ldots,n\}$ corresponding to the indexes of the initial quoted
strikes. We know that,
in any case, we must have
$0= P(0) = m_1$, so that we can define $y_1^{Ask} = 0$. Furthermore,
we know from \ref{eq:nap1} that $P(\xi)$ cannot grow at a rate faster
than $e^{-r\tau}$, so that we can define $y_n^{Ask}$ to be the
corresponding linear extrapolation of the right-most market quote
$y_{i_s}^{Ask}$, meaning $y_n^{Ask} = y_{i_s}^{Ask} + e^{-r\tau}(\xi_n
-  \xi_{i_s})$. In summary, the requirement that
the $m_i$s fall in-between the bid-ask quotes translates into $2s+2$ additional constraints, which we can
write as follows
\begin{align}
m_i &\leq y_i^{Ask}, &i \in I\cup\{1,n\}, \label{eq:fullaskconst}\\
-m_i &\leq -y_i^{Bid}, &i \in I. \label{eq:fullbidconst}
\end{align}
All previously mentioned constraints are summarized in \ref{figure:constsum}.

\subsection{The quadratic program}
Fix $N\in \N$. The choice of $N$ will be discussed in the next
Section. Let us denote by $P_N$ the estimator of the put price $P$ on $\I$
built upon the $\phi_k$'s up to level $N$ and by $e^{-r\tau}q_N$ the
corresponding inverse image by $\gamma^*$. We have explicitly, from
\ref{eq:puttocoeffs} and Theorem \ref{theorem:SVD}, item \ref{item:decomposition}),
\begin{align*}
P_N &= \gamma^* e^{-r\tau}q_N,\\
P_N &= \sum_{k=0}^N \omega_k \phi_k,\\
q_N &= e^{r\tau}\sum_{k=0}^N \lambda_k^{-1} \omega_k\psi_k,
\end{align*}
for some $\omega^T = \begin{bmatrix} \omega_0 &
  \ldots&\omega_N \end{bmatrix} \in \R^{N+1}$. Furthermore for a given matrix $M$,
we will denote by $[M]_{I,J}$ the sub-matrix obtained by extracting the
rows of $M$ at indexes in $I$ and the columns of $M$ at indexes in
$J$. When extracting all the columns, we will write
$[M]_{I,\bullet}$, and similarly for the rows. And we will naturally write
$[M]_I$ in the case where $M$ is a
vector. The SRM estimator $\omega^{\bigstar}$ is obtained
as a solution of a quadratic program. It corresponds (modulo rescaling
by the $\lambda_k$s and the discount factor) to the coefficients of
the smoothest density that verifies the no-arbitrage and bid-ask
constraints above. To that end, notice that the $\Lp_2\I$-norm of the
second derivative of $q_N$, namely $S_N =\normL{\partial_{\st}^2q_N}{\Lp_2\I}^2$,
quantifies its smoothness. $S_N$ is often used as a smoothness penalty
and has been widely used in the context of
smooth RND recovery. Obviously, the smoother $q_N$, the smaller
$S_N$. As detailed in Proposition \ref{proposition:smoothness}, $S_N$ can be
directly expressed as a quadratic form of $\omega$
involving the $N+1$ first eigenvalues of the restricted put operator
$\gamma^*$. As a consequence, $\omega^{\bigstar}$ is solution of,
\begin{align}\label{eq:optimS1bis}\tag{P1'}
&\arg\min_{\omega \in \R^{N+1}} \normL{\partial_{\st}^2q_N}{\Lp_2}^2
&\text{subject to }& 
\begin{cases} 
[P_N]_{I\cup\{1,n\}} &\leq y_{I\cup\{1,n\}}^{Ask}, \\
-[P_N]_I &\leq -y_I^{Bid},\\
A P_N &\leq b_p,\\ 
q_N(0) &= 0. \end{cases}
\end{align}
where, with a slight abuse of notations, we have written $P_N^T = \begin{bmatrix}
  P_N(\st_1)&\ldots&P_N(\st_n)\end{bmatrix}$, $y_I^{Bid}$ stands for
the vector of initial put bid quotes and $y_{I \cap \{1,n\}}^{Ask}$
stands for the vector of initial put ask quotes augmented with the no
arbitrage bounds $y_1^{Ask}=0$ and $y_n^{Ask}=
y_{i_s}^{Ask}+e^{-r\tau}(\xi_n - \xi_{i_s})$. Notice that we have added the
constraint $q_N(0)=0$, which does not arise as a natural property of
the $\psi_k$s.\\
Denote by $\phi_{0,N}(\st)^T
= \begin{bmatrix}\phi_0(\st)&\ldots&\phi_N(\st)\end{bmatrix}$ and,
similarly, write $\psi_{0,N}(\st)^T$. Then we have $[P_N]_{i} =
\phi_{0,N}(\st_i)^T\omega$ and $q_N(\st) =
\psi_{0,N}(\st)^T\Omega_N\omega$, where $\Omega_N$ is defined below in Proposition
\ref{proposition:smoothness}. Let us finally
denote by $\Phi$ the matrix whose rows are constituted by the
$\phi_{0,N}(\st_i)^T$, $i=1,\dots,n$ and write $\Phi_I = [\Phi]_{I,\bullet}$. With these notations,
\ref{eq:optimS1bis} can be rewritten in canonical form as
\begin{align}\label{eq:optimS1}\tag{P1}
&\arg\min_{\omega \in \R^{N+1}} \frac 12
\omega^T\Omega_N^4\omega  &\text{subject to } &
\begin{cases} 
\Phi_{I\cup\{1,n\}}\omega &\leq y_{I\cup\{1,n\}}^{Ask}, \\
-\Phi_I\omega &\leq -y_I^{Bid}, \\
A\Phi \omega &\leq b_p,\\
\psi_{0,N}(0)^T\Omega_N\omega &= 0.
\end{cases} 
\end{align}
which is nothing but a quadratic program in $\omega$. This result is due to the following Proposition.
\begin{proposition}\label{proposition:smoothness}
Let us write $f_N = \sum_{k=0}^N  \lambda_k^{-1} \omega_k \psi_k$ and
\begin{align}\label{eq:omegaN}
\Omega_N &=  Diag(\lambda_0^{-1}, \ldots, \lambda_N^{-1}),
\end{align}
which stands for the $(N+1)\times(N+1)$ diagonal matrix whose diagonal
entries are the $\lambda_k^{-1}$ for $k=0,\ldots,N$. Then 
\begin{align*}
\normL{\partial_{\st}^2f_N}{\Lp_2\I}^2 &= \omega^T \Omega_N^4 \omega.
\end{align*}
\end{proposition}
\begin{proof}
Notice indeed that $\partial_{\st}^2f_N = \omega^T
\Omega_N \partial_{\st}^2\psi_{0,N}$. However, as demonstrated above
in Theorem \ref{theorem:SVDbasisprop}, $\partial_{\st}^2\psi_k=
\lambda_k^{-1}\phi_k $. Hence, using the property that the $\phi_k$s
constitute an orthonormal basis of $\Lp_2\I$, we obtain
\begin{align*}
\normL{\partial_{\st}^2f_N}{\Lp_2\I}^2 &= \sum_{k=0}^N
\lambda_k^{-4}\omega_k^2 = \omega^T
\Omega_N^4 \omega.
\end{align*}
\end{proof}

\subsection{Properties of \ref{eq:optimS1} and choice of the
  spectral-cutoff $N$}
A first question that arises is whether this quadratic program eventually
admits a solution? In that perspective, it is straightforward to
notice that \ref{eq:optimS1} admits a solution if and only if
$\Span\{\phi_i, 0\leq i\leq N\}$ admits an element which satisfies the
constraints. Let us denote by $\Dom$ the subset of $\Lp_2\I$ which
satisfies the constraints described in \ref{eq:optimS1bis} and assume
that $\Dom \neq \emptyset$. Obviously, \ref{eq:optimS1} admits a
solution as soon as $N$ is large enough, since $(\phi_i)$ is complete
in $\Lp^2\I$ (see Proposition \ref{proposition:complete}). On the other hand, it
admits no solution when $\Dom = \emptyset$, that is when the
constraints are incompatible. This latter situation might result from
the presence of spurious data, since the presence of an arbitrage in the
bid-ask quotes corresponds to a real arbitrage in the market, which
would certainly be arbitraged away by practitioners.\\
A second natural question that arises, is how to choose the spectral
cutoff $N$? As detailed in \ref{eq:optimS1}, we aim at recovering the
smoothest density $q_N$ built upon $\psi_0, \ldots \psi_N $ compatible
with price quotes. As described in Theorem \ref{theorem:SVDbasisprop},
$\psi_k$ is constituted of a periodic component $h_{k,2}$ oscillating
at frequency $\rho_k/B$ around an exponential trend $h_{k,1}$, where
$\rho_k$ grows roughly speaking like $k$. It is therefore natural to
think that the smaller $N$, the smoother the singular basis functions
and thus the smoother the density $q_N$ built upon them (although this needs
not be the case, rigorously speaking). This intuitive observation, is
justified through simulations (see \ref{figure:realrnd2},
bottom graph). In practice, we therefore suggest to choose $N$ to be the smallest
$N$ such that \ref{eq:optimS1} admits a solution. This is what we
actually do in the forthcoming simulation study.\\
Finally, let us point out that we could
have chosen to impose a positivity constraint on $q_N$ at each point
of the underlying dense grid $\xi_1, \ldots, \xi_n$, as an alternative
to the in-sample convexity
constraints on the $(m_i)$s described in \ref{eq:optimS1}. However, we have
noticed via numerical simulations that results obtained in that way are less satisfying than
with the convexity constraints on the $m_i$s. We therefore opted for
the convexity constraints.

\section{Simulation study}\label{section:simuls}
We run a simulation study both on real and simulated data. The purpose of the
estimation on simulated data is mostly to show that the SRM returns a
valid RND estimator in extreme cases, when as little as $5$
market quotes are available.\\
Recall from Lemma \ref{lemma:thetak} that, from $k = 1$ onward, we can write
$\beta_k \approx 2e^{-\frac{\pi}2 - k\pi}$ in \ref{eq:defrhok} above.
This approximation is not valid for $k=0$. In that case, however, we can solve
\ref{eq:tm} numerically to obtain $\rho_0 = 1.875104069$. This is the
value of $\rho_0$ we use in the following simulation study.
\begin{table}[hbt]
\caption{S\&P 500 put option prices, Jan. 5, 2005. S\&P 500 Index
  closing level $= 1183.74$; Option expiration $= 03/18/2005$ ($72$ days);
  $r = 2.69\%$; $\delta=1.70\%$.}\label{table:realdata}
\begin{small}
\begin{tabular}{r | c c c c c c c c c c}
\hline
Strike price &500 &550 &600 &700 &750 &800 &825 &850 &900 &925\\
Best bid     &0.00 &0.00 &0.00 &0.00 &0.00 &0.10 &0.00 &0.00 &0.00 &0.20\\ 
Best offer   &0.05 &0.05 &0.05 &0.10 &0.15 &0.20 &0.25 &0.50 &0.50 &0.70\\\hline
Strike price &950 &975 &995 &1005 &1025 &1050 &1075 &1100 &1125 &1150\\
Best bid     &0.50 &0.85 &1.30 &1.50 &2.05 &3.00 &4.50 &6.80 &10.10 &15.60\\
Best offer   &1.00 &1.35 &1.80 &2.00 &2.75 &3.50 &5.30 &7.80 &11.50 &17.20\\\hline
Strike price &1170 &1175 &1180 &1190 &1200 &1205 &1210 &1215 &1220 &1225\\
Best bid     &21.70 &23.50 &25.60 &30.30 &35.60 &38.40 &41.40 &44.60
&47.70 & 51.40\\ 
Best offer   &23.70 &25.50 &27.60 &32.30 &37.60 &40.40 &43.40 &46.60
&49.70 &53.40\\\hline 
Strike price &1250 &1275 &1300 &1325 &1350 & & & & & \\
Best bid     &70.70 &92.80 &116.40 &140.80 &165.50 & & & & & \\
Best offer   &72.70 &94.80 &118.40 &142.80 &167.50 & & & & & \\\hline
\end{tabular}
\end{small}
\end{table}  

\subsection{Real data}

We use the bid ask quotes reported in \mycite[Table~1]{Figlewski2008} for
put options on the S\&P 500 Index on January 5, 2005. For
completeness, we reproduce the table here in \ref{table:realdata}. We
choose $B = 2*S_0e^{(r-\delta)\tau}$, which corresponds
to two times the Forward price on the underlying stock. This choice is
arbitrary and produces an interval $\I$, which is symmetric around the
forward price. We observe from our simulation that the result is largely
independent of the choice of $B$. However, the higher $B$, the
higher we will need to go into the spectrum of $\gamma^*$, since the
smoothest RND that fits the data will be more and more
concentrated around the center of the interval $\I$. As regards the constraints, we
choose the grid $\st_1,\ldots, \st_n$ to be such that $\st_k = k-1,
k=1,\ldots,\PE{B}+1$ and if $\PE{B} < B$, we add $\st_{\PE{B}+2} = B$. Of
course, this grid contains the initial $35$ quoted strike prices since they are
integer valued. With the above
choice of $B$,
the quadratic program given in \ref{eq:optimS1} finds a feasible solution from
spectral cutoff $66$ onward. We report $q^{\bigstar}_{66}$
below in \ref{figure:realrnd2}. For the sake of comparison, we
plot on the same figure the log-normal
distribution obtained by least-square fit to the put prices obtained
as average of the bid-ask quotes. The only parameter of the log-normal
distribution that must be fitted is $\sigma$ (see Proposition \ref{proposition:BS}), and we find
$\sigma_{opt} = 0.143$. Interestingly, $q^{\bigstar}_{66}$
displays a small bump at the beginning of its left-tail, which does not appear in
\mycite[Fig.~8]{Figlewski2008} and could hardly be accounted for by
parametric methods. Notice the small blip next to
$B$ in \ref{figure:realrnd2}. This boundary effect is due to the fact that all the $\psi_k$s and their first
derivative are worth $0$ in $B$. In order to show that the choice of $B$ has very
little impact, we compute the RND estimator for $B =
1.4*S_0e^{(r-\delta)\tau}$. Results are reported in
\ref{figure:realrnd}. As was expected, first feasible points appear at
much lower spectral cutoffs, namely from spectral cutoff $26$
onward. Therefore, we plot $q_{26}^{\bigstar}$. As can be seen from \ref{figure:realput}, the put
prices $P_{26}^{\bigstar}$ arising from \ref{eq:optimS1} lie inside the
bid ask quotes, while the ones produced by the fitted log-normal
density lie outside.
\begin{figure}[H]
\includegraphics[width=\textwidth, height =
0.4\textheight]{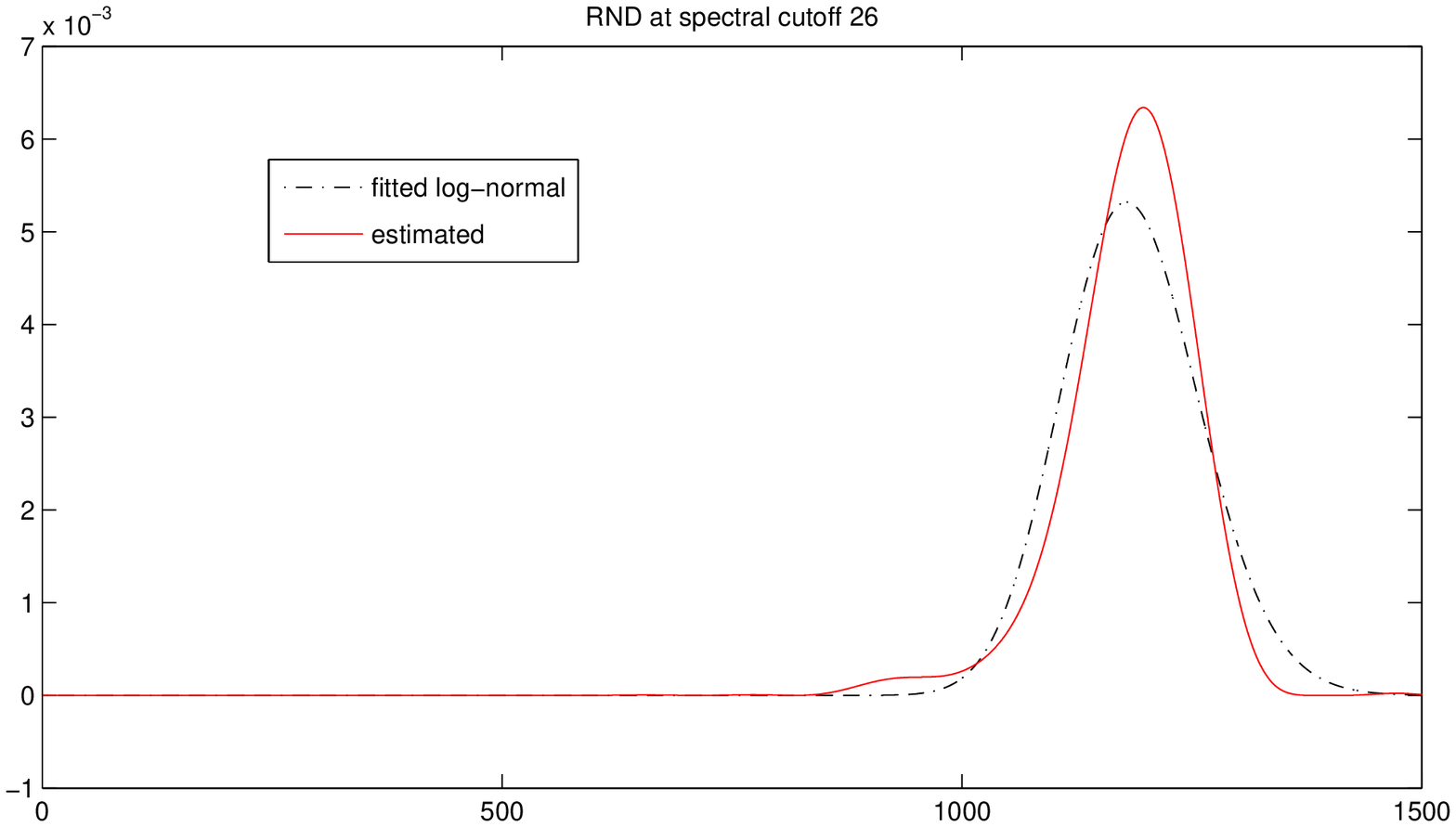}
\includegraphics[width=\textwidth, height =
0.4\textheight]{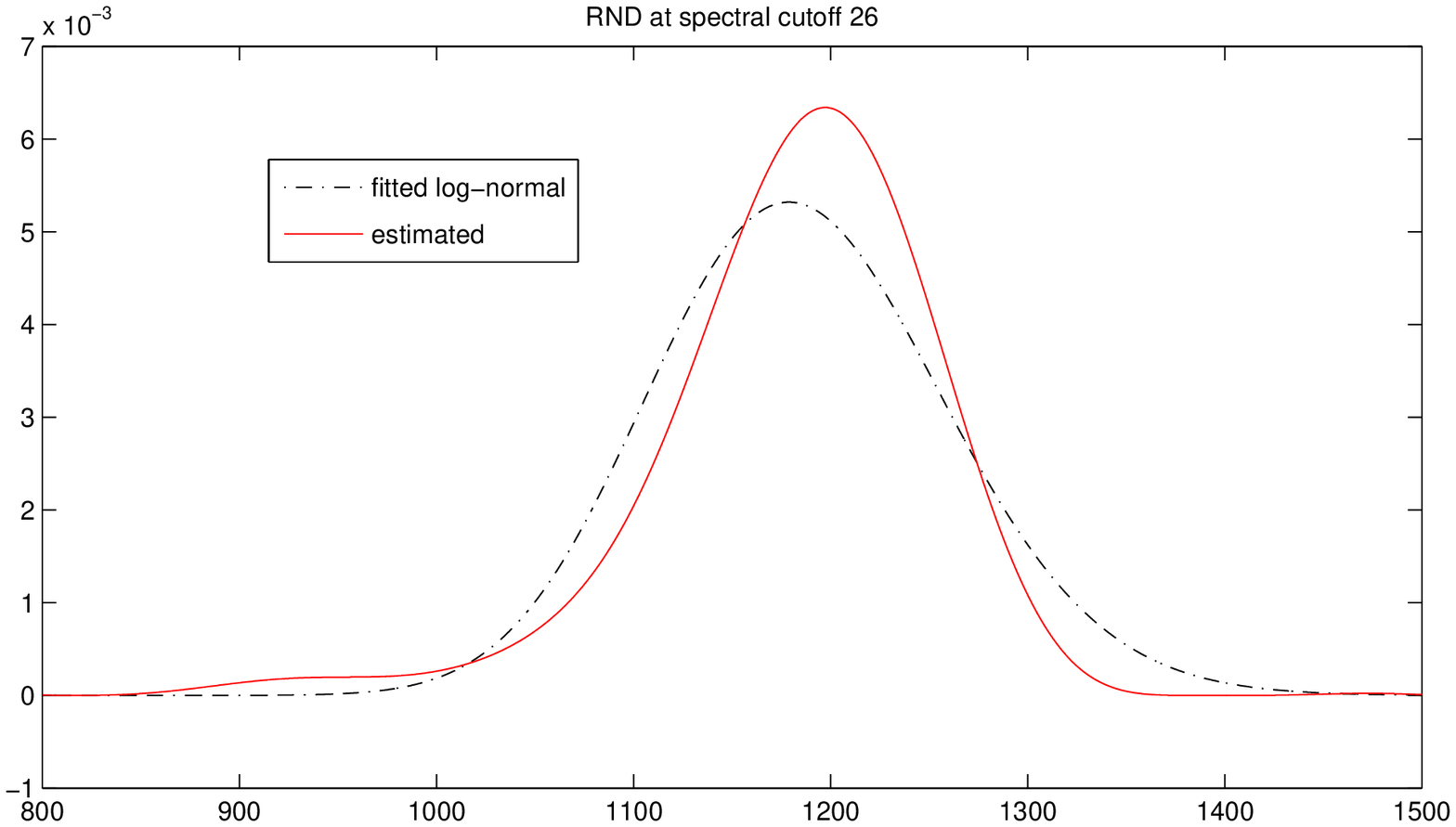}
\caption{Here we plot the RND $q^{\bigstar}_{26}$ (solid line) estimated from the real price quotes
  reported in \ref{table:realdata}. We choose $B = 1.4*F_0 =
  1.4*S_0*e^{(r-\delta)\tau} = 1660$ for that plot. In addition, we
  plot the best log-normal fit (in a least-square sense) to the average
  price quotes (dashed line). It is obtained for $\sigma_{opt} =
  0.143$. At the top, we display the full left tail
  of the RND $q^{\bigstar}_{26}$. At the bottom, we zoom in on the fat
  left tail of the estimated RND distribution.}\label{figure:realrnd}
\end{figure}
\begin{figure}[H]
\includegraphics[width=\textwidth, height =
0.4\textheight]{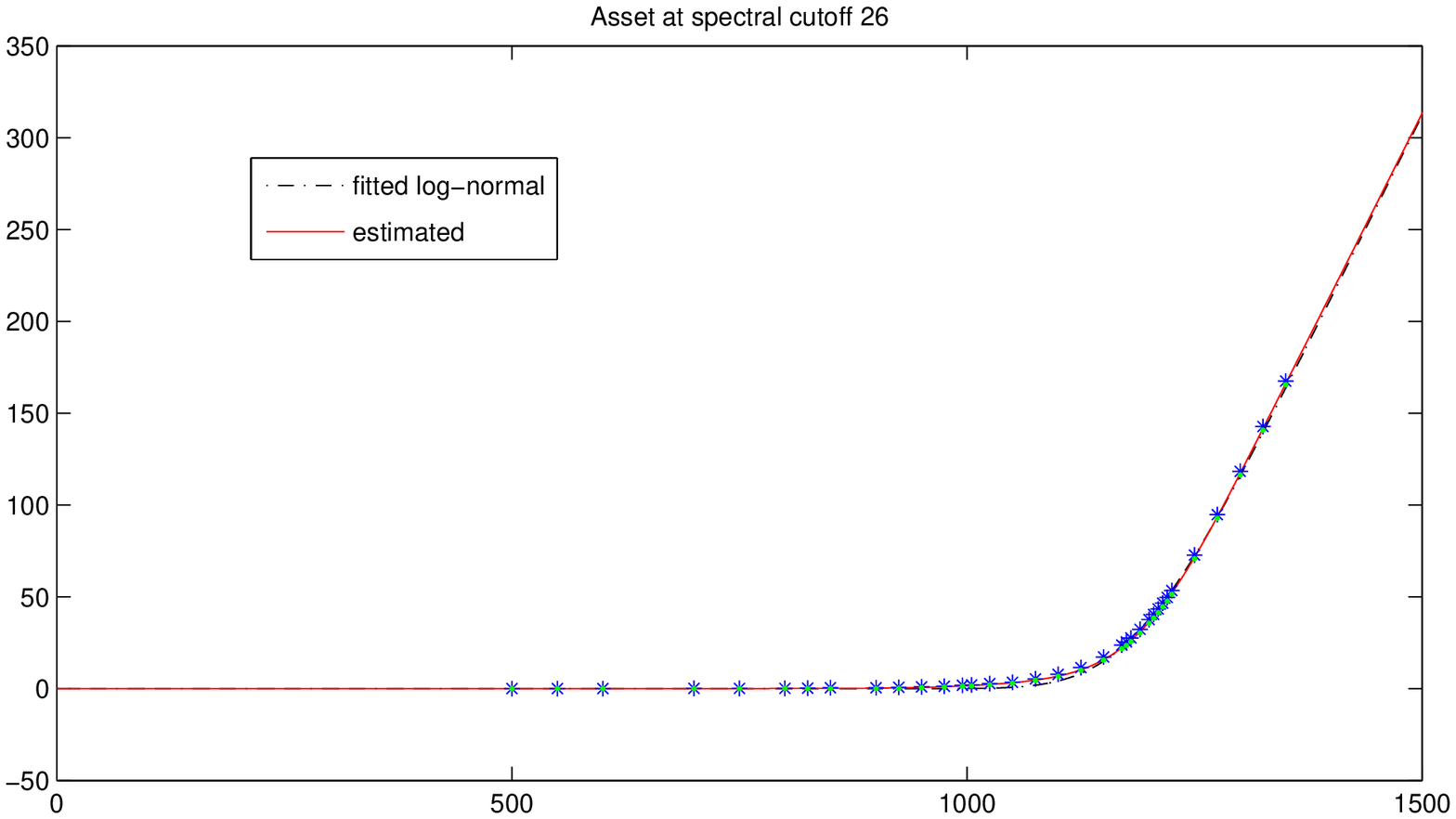}
\includegraphics[width=\textwidth, height =
0.4\textheight]{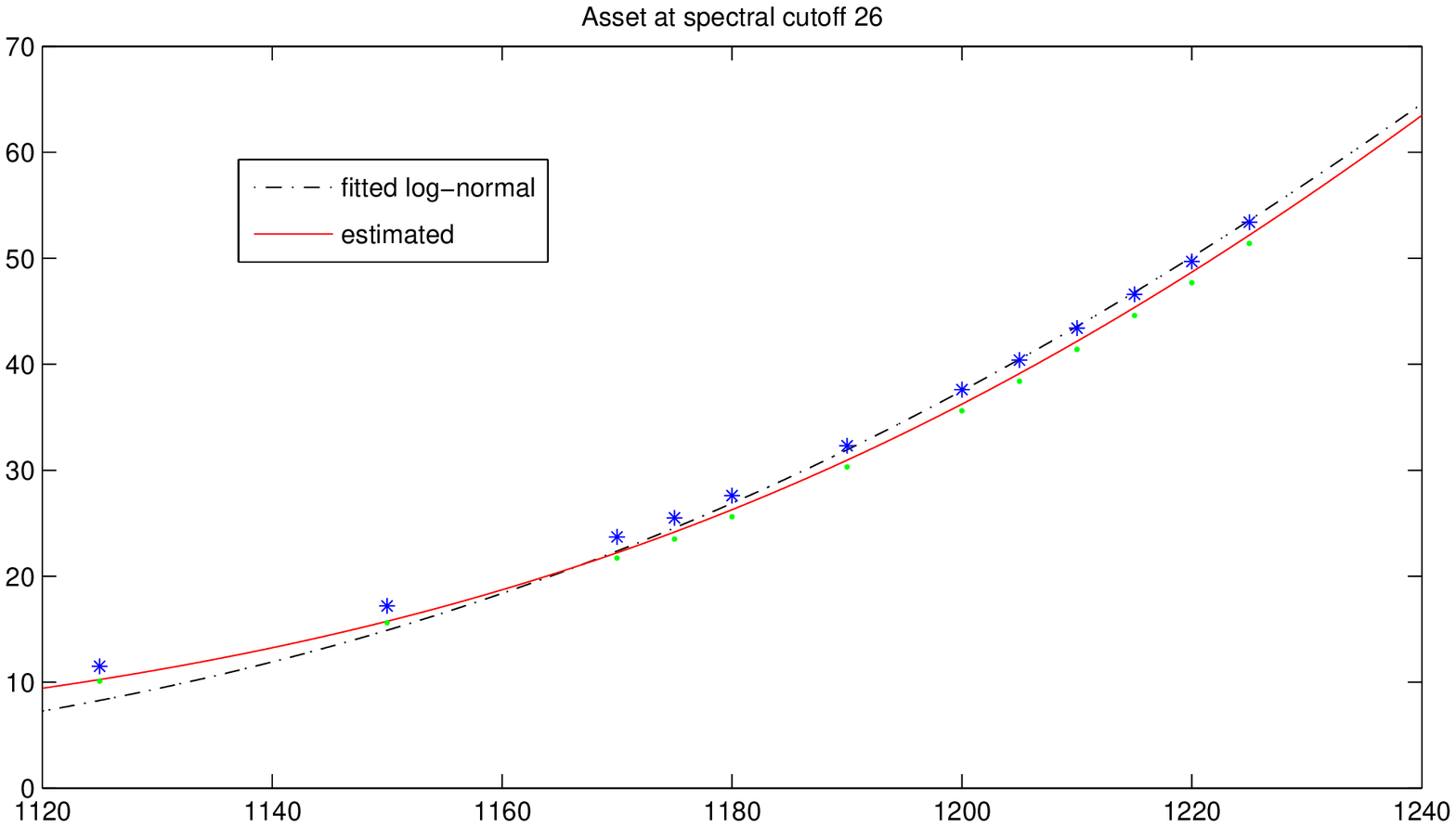}
\caption{Here we plot the fitted put prices obtained from the setting
  described above in \ref{figure:realrnd}. The solid line corresponds
  to the fitted prices $P^{\bigstar}_{26}$, while the dashed line
  corresponds to the fitted prices obtained from a log-normal
  distribution. The stars and dots correspond to market ask and bid
  quotes, respectively. At the top, we give a large view of the fits. At the
  bottom we zoom in to show that $P^{\bigstar}_{26}$ lies inside the
  market quotes, while the fitted log-normal prices lie outside.}\label{figure:realput}
\end{figure}
\begin{figure}[H]
\includegraphics[width=\textwidth, height =
0.4\textheight]{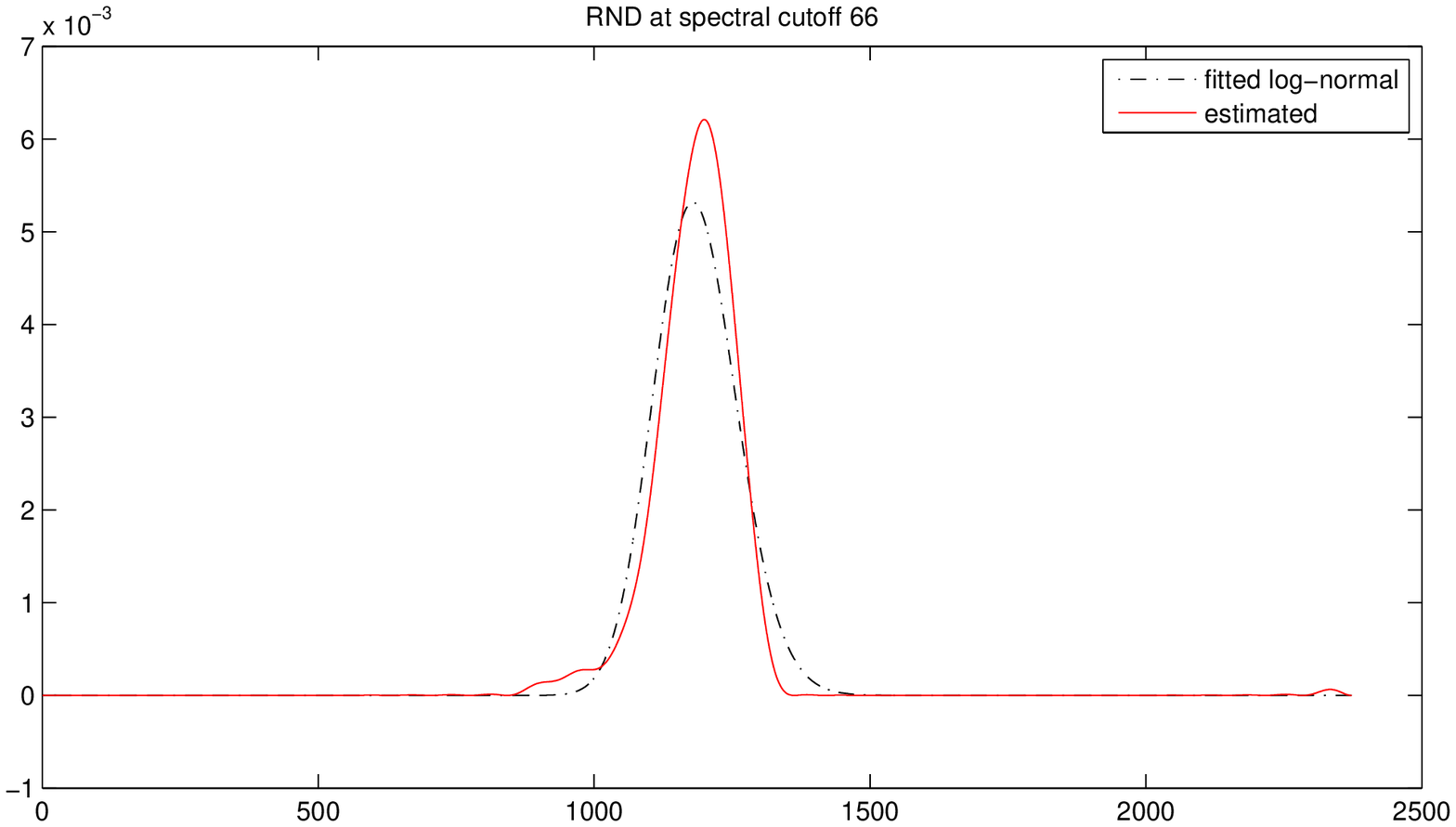}
\includegraphics[width=\textwidth, height =
0.4\textheight]{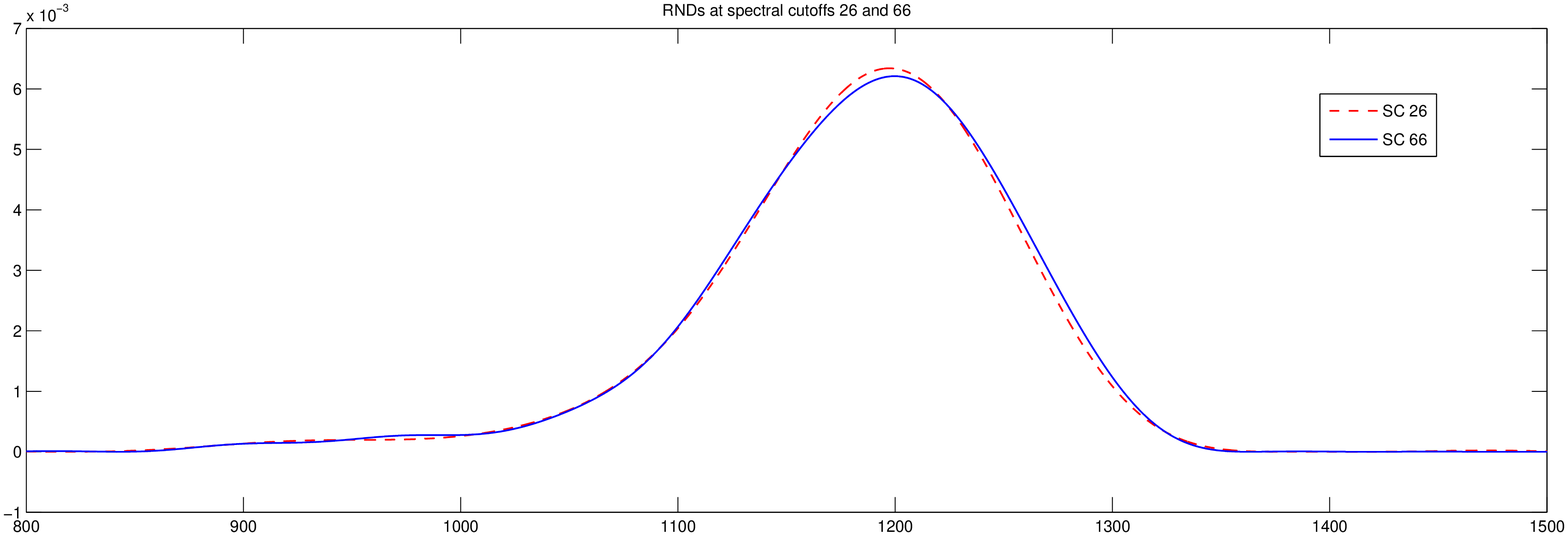}
\caption{Here we plot the RND $q^{\bigstar}_{66}$ (solid line) estimated from the real price quotes
  reported in \ref{table:realdata}. We choose $B = 2*F_0 =
  2*S_0*e^{(r-\delta)\tau} = 2372$ for that plot. In addition, we
  plot the best log-normal fit (in a least-square sense) to the average
  price quotes (dashed line). It is obtained for $\sigma_{opt} =
  0.143$. At the top, we display the full left tail
  of the RND $q^{\bigstar}_{66}$ and its full right tail up to $B$. At
  the bottom, we superimpose $q^{\bigstar}_{66}$ (solid line) with
  $q^{\bigstar}_{26}$ (dashed line) obtained in \ref{figure:realrnd} for an other
  choice of $B$. Notice the strong agreement between both densities,
  which highlights the stability of the SRM with respect to the choice
  of $B$. Interestingly, $q^{\bigstar}_{66}$ is slightly more bumpy
  than $q^{\bigstar}_{26}$ at the level of its left fat-tail. This
  reinforces our argument that smoothness goes hand in hand with low
  spectral cutoff.} \label{figure:realrnd2}
\end{figure}

\subsection{Simulated data}
As regards the simulated data, we work in the Black-Scholes
setting. In that context the price of a put option admits a closed form solution
and the RND is log-normal (see Proposition \ref{proposition:BS}). We model the
bid-ask spread as a random noise around the true price given by the
Black-Scholes formula. More precisely, for a given set of quoted
strikes $\st_1<\ldots < \st_s$ and corresponding put prices
$P(\st_1),\ldots, P(\st_s)$, we write $y_i^{Ask} = P(\st_i)+ z_i/2$ and
$y_i^{Bid} = P(\st_i) - z_i/2$, where $z_i = \max(1, \min(3,
\varpi \abs{\xi_i}))$, the $\xi_i$'s are iid standard normal random
variables and $\varpi = 0.1 \max_{1\leq i \leq s} P(\st_i)$. The bounds
$1$ and $3$ are chosen by analogy with the real data quotes in \ref{table:realdata}. Of course,
the bid-ask quotes we obtain in that way are not
arbitrage free. However, they contain the true put price $P(\st)$,
which, given the nature of the quadratic program described in \ref{eq:optimS1} above, is all
that matters to approximate the true RND. For the sake of simplicity,
we choose $r=0$, $\delta=0$, $\tau=1$, $S_0=100$, and $\sigma=0.3$ and
$B=2*F_0 = 2*S_0$. In addition we set a first strike price at
$\PE{F_0}$ and spread the others on its left and right sides at unit
length distance away from each other until we obtain $s$
strikes. More precisely, the second strike would be $\PE{F_0}-1$, the
third $\PE{F_0} +1$, the fourth $\PE{F_0}-2$ and so on and so
forth. We plot the results for the first two spectral cutoffs at which a
feasible point is found below in \ref{figure:LNrnd5} in the case where there are as little as $s=5$
bid ask quotes and in \ref{figure:LNrnd50} in the case where there are as
many as $s=50$ of them. In any case, we can see that we obtain a
smooth density that resembles the log-normal density generating the
initial quoted prices and that the estimate is stable from one
spectral cutoff to another. Of course, the more strikes we have, the
better the fit. Besides, we observe as expected from an other
simulation not reported here that, the smaller the bid-ask spread, the
better the fit. Notice
once again that the fitted right-tail reaches zero in $B$, while the
true one is strictly positive at that point. As before, this is due to
the fact that $\psi_k(B)=0$.
\begin{figure}[H]
\includegraphics[width=\textwidth, height =
0.4\textheight]{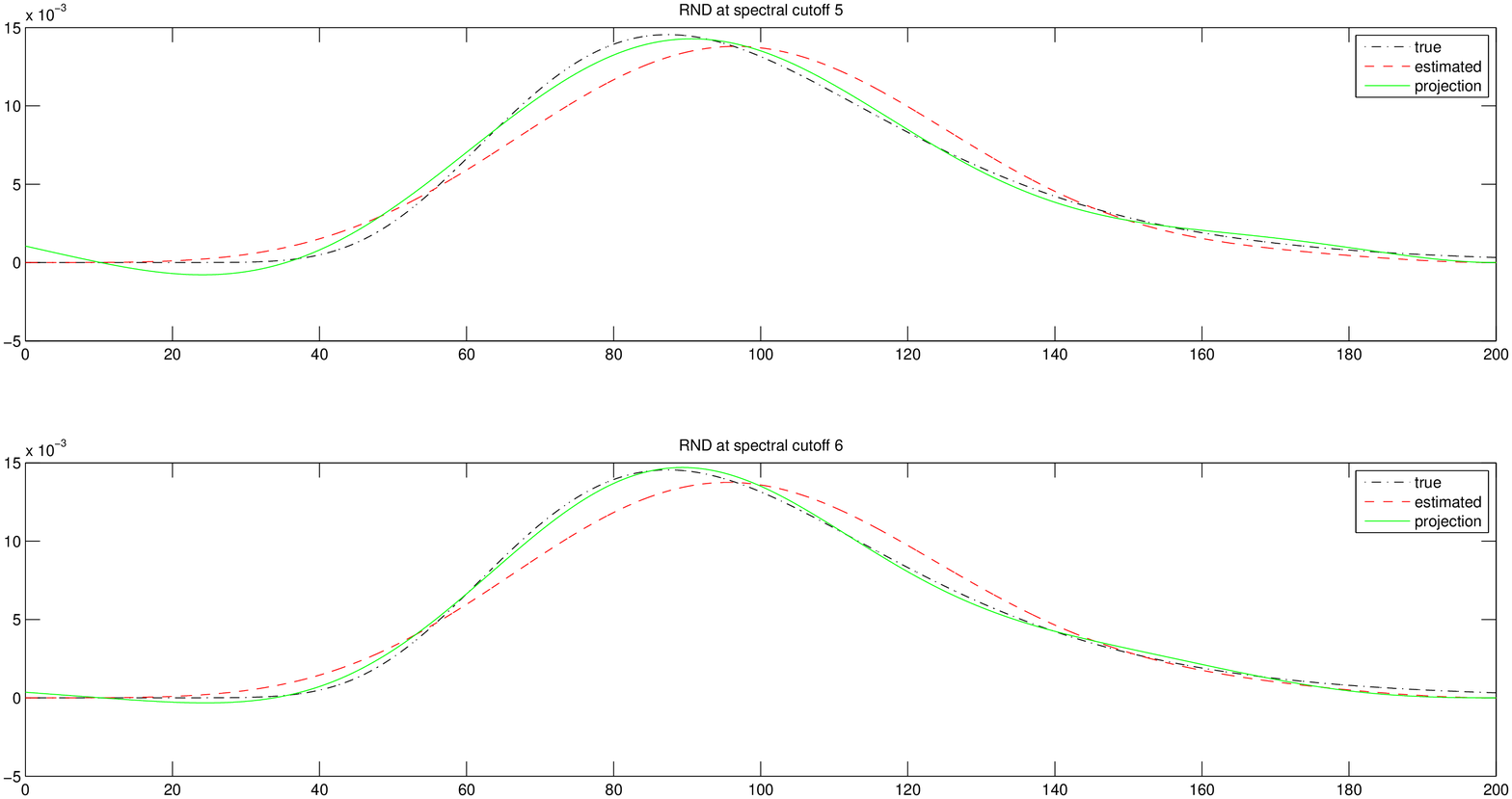}
\includegraphics[width=\textwidth, height =
0.4\textheight]{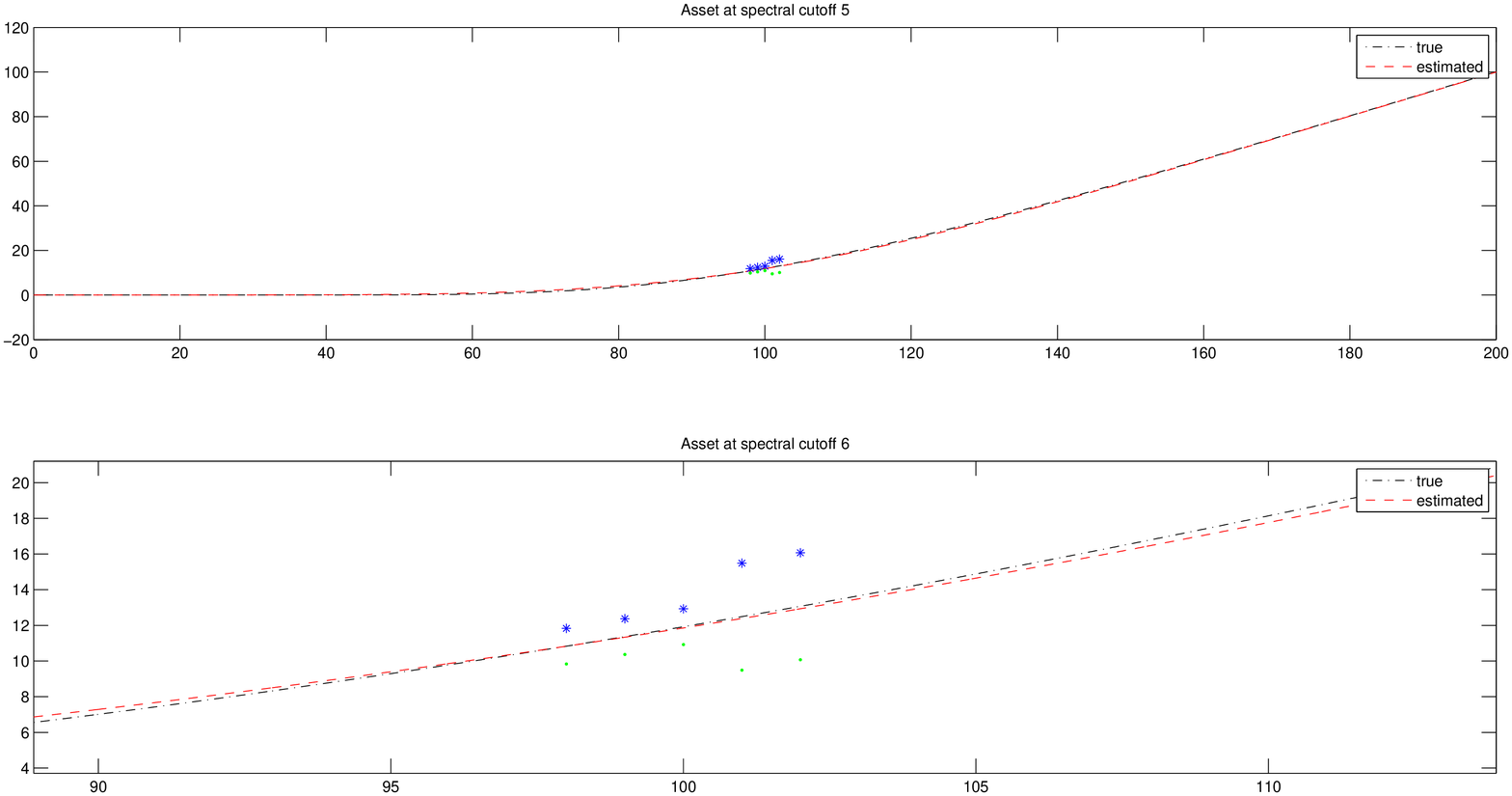}
\caption{Here we are in the case of $5$ simulated bid ask quotes and
  with $B = 2*F_0 = 200$. The
  first two plots display $q^{\bigstar}_5$ and $q^{\bigstar}_6$ (dashed
line), the true log-normal RND used to generate the prices
(dashed-dotted line) and the orthogonal projection of the true
log-normal RND on $\{\psi_0, \dots, \psi_N\}$ for $N=5$ and $N=6$ (solid
line), respectively. The last two plots display the fitted put
prices, that is $P_5^{\bigstar}$ and $P_{6}^{\bigstar}$ (dashed line)
together with the true prices (dashed-dotted line).}\label{figure:LNrnd5}
\end{figure}
\begin{figure}[H]
\includegraphics[width=\textwidth, height =
0.45\textheight]{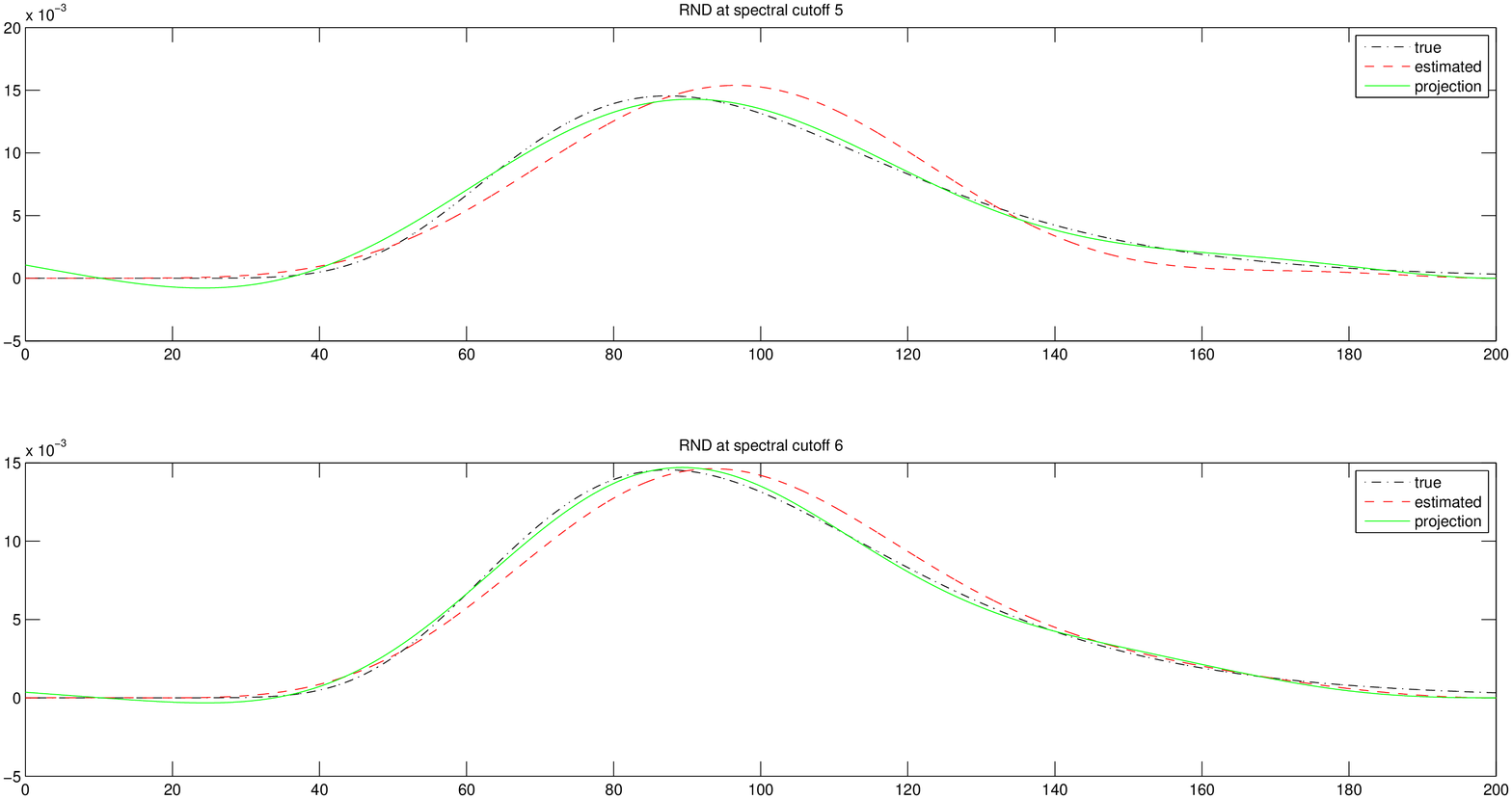}
\includegraphics[width=\textwidth, height =
0.45\textheight]{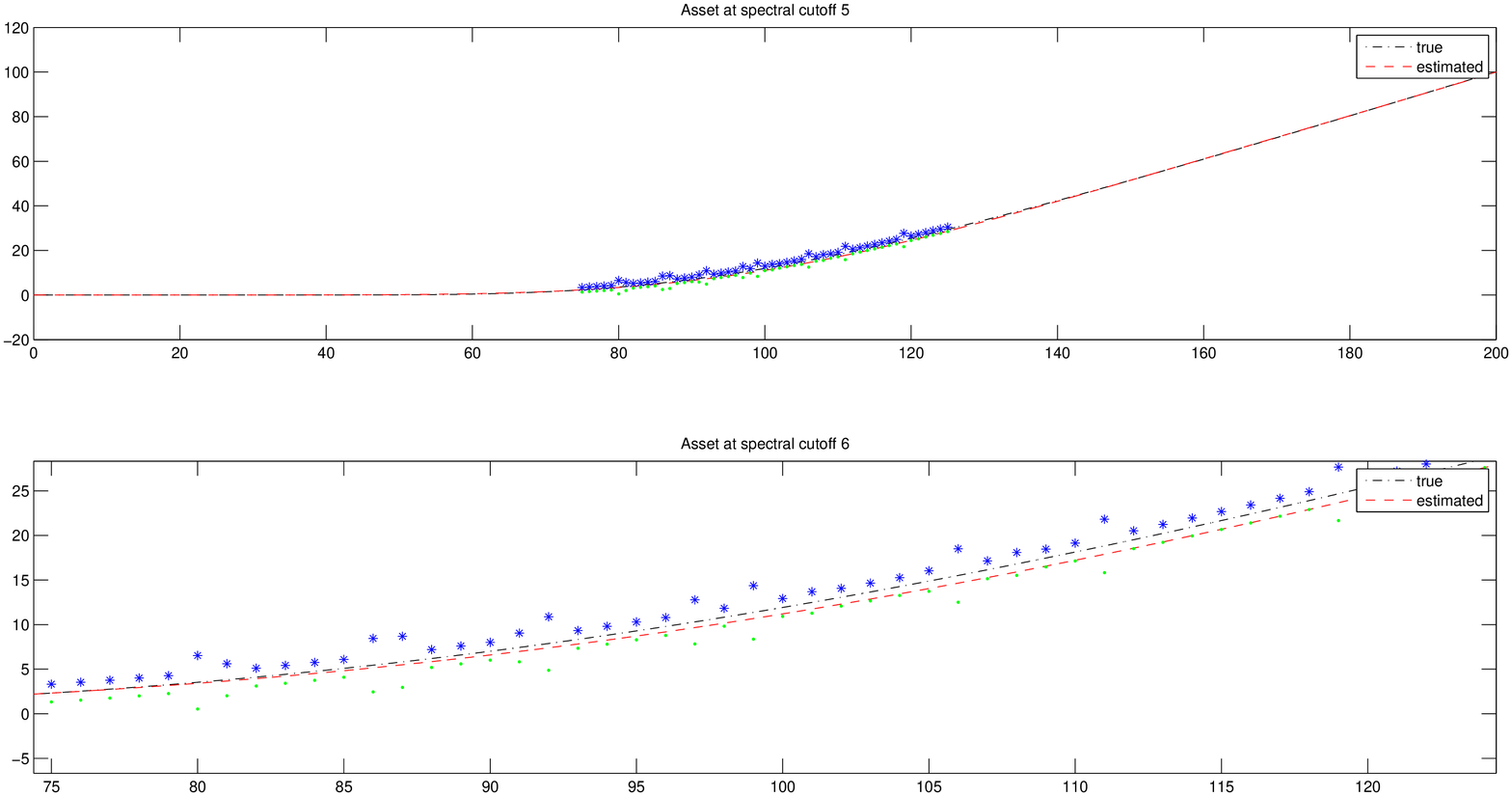}
\caption{Here, we repeat the same plots as in \ref{figure:LNrnd50} in
  the case of $50$ simulated bid-ask quotes.}\label{figure:LNrnd50}
\end{figure}

\section*{Appendix}

\subsection*{Refresher on the Black-Scholes model}
This is a well-known result of mathematical finance.
\begin{proposition}\label{proposition:BS}
Let us denote by $S_0$ the price today of a stock paying dividends
continuously over time at a
constant rate $\delta$ and by $r$ the continuously compounded
risk-free rate. The arbitrage price today of a put option on that
stock maturing at time $\tau$ is given by the following closed form formula,
\begin{align*}
P(\st) &= \st e^{-r\tau} \mathcal{N}(-d_2) - S_0e^{-\delta\tau}\mathcal{N}(-d_1),
\end{align*}
with
\begin{align*}
d_1 &= \frac{\ln(S_0/\st) + [(r-\delta) + \frac 12 \sigma^2]
  \tau}{\sigma\sqrt{\tau}},
& d_2 &= d_1 - \sigma \sqrt{\tau},
\end{align*}
where $\sigma$ stands for the volatility of the stock and
$\mathcal N$ for the standard normal cumulative distribution. In
addition, the RND is log-normal and writes as
\begin{align*}
q(x) &= \frac{1}{\sqrt{2\pi} \sigma \tau x}\exp\left( -
  \frac{[\ln(x/S_0) - (r - \delta) \tau +\frac 12 \sigma^2 \tau]^2}{2\sigma^2\tau}   \right).
\end{align*}
\end{proposition}
\begin{proof}
These results can be found in see \mycite[p.117]{Musiela2008}, for example.
\end{proof}

\subsection*{Additional results relative to $\gamma$ and $\gamma^*$}
We now present three results relative to $\gamma$ and $\gamma^*$,
which are either used in the core of the paper or of interest in their
own right.

\begin{proposition}
The operators $\gamma$ and $\gamma^*$ admit
no eigenvectors.
\end{proposition}
\begin{proof}
Suppose $f$ is an eigenvector of $\gamma$ associated to eigenvalue
$\lambda$, then denote by 
\begin{align}\label{eq:breve}
\breve f(t) = f(B-\st),
\end{align}
and notice that for
all $\st\in \I$, a direct application of Lemma \ref{lemma:breve} allows to write
\begin{align*}
\lambda \breve f (B-\st) = \lambda f(\st) &= \gamma f(\st) = \gamma^*\breve f(B-\st).
\end{align*}
Thus $\breve f$ must be an eigenvector of $\gamma^*$. However, it is
well known that $\gamma^*$ admits no eigenvalue since, for any
$\lambda\neq 0$,
\begin{align*}
\lambda f(\st)&=  \gamma^* f(\st) = \int_0^{\st} \theta^*(\st,x)f(x)dx, &\st\in\I,
\end{align*}
defines a \textsl{homogeneous Volterra equation} in $f$, whose unique trivial solution is
$f=0$ (see \mycite[p.239, Th.~5.5.2]{Debnath1990}).
\end{proof}

Finally, let us point out the two following useful lemmas.
\begin{lemma}\label{lemma:partial4}
Let us denote by $\partial_{\st}^k$ the $k^{th}$ order partial
differential operator with
respect to $\st$. Then, for any $f\in\Lp_2\I$, we have the following results.
\begin{align*}
f &= \partial_{\st}^2\gamma f,  &f &= \partial_{\st}^2 \gamma^* f,\\
f &= \partial_{\st}^4\gamma^*\gamma f,& f
&= \partial_{\st}^4\gamma\gamma^* f.
\end{align*}
\end{lemma}
\begin{proof}
Notice indeed that
\begin{align*}
\partial_{\st} \gamma f(\st) &= \partial_{\st} \int_{\st}^B (x-\st)f(x)dx =- \int_{\st}^B
f(x)dx, \\
\partial_{\st} \gamma^* f(\st) &= \partial_{\st} \int_0^{\st} (\st-x)f(x)dx =\int_0^{\st}
f(x)dx. 
\end{align*}
Therefore, we obtain immediately
\begin{align*}
f = \partial_{\st}^2 \gamma f = \partial_{\st}^2 \gamma^*f.
\end{align*}
The remaining of the proof follows directly from these first
results. Notice indeed that,
\begin{align*}
\partial_{\st}^4\gamma^*\gamma f
= \partial_{\st}^2 [\partial_{\st}^2\gamma^*](\gamma f)
= \partial_{\st}^2\gamma f = f.
\end{align*}
which concludes the proof.
\end{proof}

\begin{lemma}\label{lemma:breve}
For any $f \in \Lp_2\I$ and $\st\in \I$, we have $\gamma f(\st) = \gamma^*
\breve f(B-\st)$ (see \ref{eq:breve} for notations).
\end{lemma}
\begin{proof}
Perform the change of variable $u = B-x$ to obtain
\begin{align*}
\gamma f(\st) &= \int_{\st}^B (x-\st)f(x)dx\\
&= \int_0^{B-\st}([B-\st] - u) \breve f(u)du = \gamma^*\breve f(B-\st).
\end{align*}
\end{proof}

\subsection*{Relation between the $(\phi_k)$s and the $(\psi_k)$s}

We believe that $m_0(q)$ and $m_1(q)$ could be readily
estimated from the data, so that \ref{eq:calltocoeffs} could be used
to construct a second estimator of the RND based on the restricted
call operator. This second estimator could eventually be combined with
the one obtained from the SRM above. To that
end, and for the sake of completeness, we compute the scalar products between elements
of the two singular bases. Results are reported in the following
proposition. 
\begin{proposition}\label{proposition:crossscalar}
Let us write
\begin{align*}
\mathfrak{p}_{k,m}(x,y) &= (-x^3 + x^2y)(-1)^{m+k} - xy^2 + y^3,\\
\mathfrak{q}_{k,m}(x,y) &= (x^3+x^2y)(-1)^k + (y^3 + y^2x)(-1)^m.
\end{align*}
Then, we have the following relationships,
\small{
\begin{align*}
&\ms{\phi_k}{\psi_m}\\
&= 4 \frac{\mathfrak{p}_{k,m}(\rho_k,
  \rho_m)e^{-\rho_k-\rho_m} - \mathfrak{q}_{k,m}(\rho_k, \rho_m)e^{-\rho_k} +
  \mathfrak{q}_{k,m}(\rho_m,\rho_k)e^{-\rho_m} +
  \mathfrak{p}_{k,m}(\rho_m,\rho_k)}{(\rho_k^4 - \rho_m^4)(1 + (-1)^m e^{-\rho_m})(1 + (-1)^k
e^{-\rho_k})}, k\neq m,
\end{align*}}
\begin{align*}
\ms{\phi_k}{\psi_k}&= \frac{-e^{-2\rho_k}(\rho_k+2) + 2\rho_k(-1)^ke^{-\rho_k}
  - \rho_k +2}{(e^{-\rho_k} + (-1)^k)^2 \rho_k}.
\end{align*}
On the way, we obtain,
\small{\begin{align*}
&\ms{h_{k,1}}{h_{m,1}}\\
&= ((-1)^k+ (-1)^m) \frac{(\rho_k+\rho_m)(
  e^{-\rho_m} - e^{-\rho_k}) + (-1)^k (\rho_k-\rho_m) (1 - e^{-(\rho_k
    + \rho_m)})}{(\rho_k^2-\rho_m^2) (1 + (-1)^ke^{-\rho_k})(1+
  (-1)^me^{-\rho_m})}, k\neq m,\\
&\ms{h_{k,1}}{h_{m,2}}\\
&= ((-1)^k -
(-1)^m) \frac{(\rho_k +\rho_m)(e^{-\rho_m} + e^{- \rho_k}) - (-1)^k (\rho_k - \rho_m)(1
  + e^{-(\rho_k+ \rho_m)}) }{(\rho_k^2+\rho_m^2)(1 + (-1)^m e^{-\rho_m})(1 + (-1)^k
e^{-\rho_k})}, k\neq m,
\end{align*}}
\begin{align*}
\ms{h_{k,1}}{h_{k,1}} &= \frac{ 1 - e^{-2\rho_k} +
  2(-1)^k\rho_ke^{-\rho_k}}{\rho_k((-1)^k + e^{-\rho_k} )^2},\\
\ms{h_{k,1}}{h_{k,2}}&= 0,\\
\ms{h_{k,2}}{h_{m,2}}&= \delta_{k,m} - \ms{h_{k,1}}{h_{m,1}}.
\end{align*}
\end{proposition}

\begin{proof}
Recall that, for all $k,m$, we have defined
\begin{align*}
h_{k,1} &=  a_{k,1}f_{k,1} + a_{k,2}f_{k,2}, &h_{k,2} &=  a_{k,3}f_{k,3} + a_{k,4}f_{k,4},\\
\phi_k &= h_{k,1} + h_{k,2}, & \psi_k &= h_{k,1} -h_{k,2}.
\end{align*}
Besides, we have that
\begin{align*}
\ms{\phi_k}{\phi_{m}} &= \delta_{k,m} = \ms{h_{k,1}}{h_{m,1}} +
\ms{h_{k,2}}{h_{m,2}} + \ms{h_{k,1}}{h_{m,2}} +
\ms{h_{k,2}}{h_{m,1}},\\
\ms{\psi_k}{\psi_{m}} &= \delta_{k,m} = \ms{h_{k,1}}{h_{m,1}} +
\ms{h_{k,2}}{h_{m,2}} - \ms{h_{k,1}}{h_{m,2}} - \ms{h_{k,2}}{h_{m,1}}.
\end{align*}
Therefore, we obtain the following relationships,
\begin{align*}
\delta_{k,m} &= \ms{h_{k,1}}{h_{m,1}} +
\ms{h_{k,2}}{h_{m,2}},\\
0 &= \ms{h_{k,1}}{h_{m,2}} + \ms{h_{k,2}}{h_{m,1}}.
\end{align*}
Which leads to
\begin{align*}
\ms{\phi_k}{\psi_m} &= \ms{h_{k,1}}{h_{m,1}} -
\ms{h_{k,2}}{h_{m,2}} - \ms{h_{k,1}}{h_{m,2}} +
\ms{h_{k,2}}{h_{m,1}},\\
&= 2 (\ms{h_{k,1}}{h_{m,1}} - \ms{h_{k,1}}{h_{m,2}})  - \delta_{k,m}.
\end{align*}
Now, it remains to compute $\ms{h_{k,1}}{h_{m,1}}$ and
$\ms{h_{k,1}}{h_{m,2}}$. The results follow from lengthy and tedious but straightforward
computations and are therefore not reported here. 
\end{proof}

\subsection*{From the RND $q$ of $S_{\tau}$ to the density of $\ln
  S_{\tau}$}
Some authors have chosen to focus on the estimation of the density of
$\log S_\tau$ rather than on the density of $S_{\tau}$ itself. Both
densities relate by a simple transformation, as described in the
following proposition. In our case, this transformation can be readily
applied since the SRM returns an analytic expression for the
estimated RND.
\begin{proposition}
If $X$ admits $f(x)$ for density on $\R$, then $Y = \exp(X)$ admits $\tfrac 1y
f(\ln y)$ for density on $\R^+$. Conversely, if $Y$ admits $f(y)$ for
density on $\R^+$, then $X = \ln(Y)$ admits $e^x f(e^x)$ for density
on $\R$. 
\end{proposition}

\section*{Acknowledgments}
The author is deeply grateful to Peter Tankov for his careful
reading of this manuscript and for his constructive and insightful
comments, which greatly contributed to improve its clarity and
content. The author is of course solely responsible for any eventual
remaining error. Finally, the author would like to acknowledge
interesting conversations with G\'erard Kerkyacharian and Dominique Picard.


\bibliographystyle{plain}


\begin{thebibliography}{10}

\bibitem{Abken1996}
{\sc Peter~A. Abken, Dilip~B. Madan, and Sailesh Ramamurtie}, {\em Estimation
  of risk-neutral densities by {H}ermite polynomial approximation: with an
  application to eurodollar futures options}, Federal Reserve Bank of Atlanta,
  working paper no. 96-5,  (1996).

\bibitem{Ait-Sahalia2003}
{\sc Yacine A\"{i}t-Sahalia and Jefferson Duarte}, {\em Nonparametric option
  pricing under shape restrictions}, J. Econometrics, 116 (2003), pp.~9--47.

\bibitem{Bahra1997}
{\sc Bhupinder Bahra}, {\em Implied risk-neutral probability density functions
  from option prices: theory and application}, Bank of England, working paper
  no. 1368-5562,  (1997).

\bibitem{Black1973}
{\sc Fischer Black and Myron Scholes}, {\em The pricing of options and
  corporate liabilities}, J. Polit. Econ., 81 (1973), pp.~637--654.

\bibitem{Bondarenko2003}
{\sc Oleg Bondarenko}, {\em Estimation of risk-neutral densities using positive
  convolution approximation}, J. Econometrics, 116 (2003), pp.~85--112.

\bibitem{Breeden1978}
{\sc Douglas~T. Breeden and Robert~H. Litzenberger}, {\em Prices of
  state-contingent claims implicit in option prices}, J. Bus., 51 (1978),
  pp.~621--651.

\bibitem{Bu2007}
{\sc Ruijun Bu and Kaddour Hadri}, {\em Estimating option implied risk-neutral
  densities using spline and hypergeometric functions}, Econometrics J., 10
  (2007), pp.~216--244.

\bibitem{Buchen1996}
{\sc Peter~W. Buchen and Michael Kelly}, {\em The maximum entropy distribution
  of an asset inferred from option prices}, J. Finan. Quant. Anal., 31 (1996),
  pp.~143--159.

\bibitem{Cont1997}
{\sc Rama Cont}, {\em Beyond implied volatility: extracting information from
  option prices}, in Econophysics: an emergent science, J.~Kert\'{e}sz and
  I.~Kondor, eds., Dordrecht, Kluwer, 1997.

\bibitem{Cox1976}
{\sc John~C. Cox and Stephen~A. Ross}, {\em The valuation of options for
  alternative stochastic processes}, Journal of Financial Economics, 3 (1976),
  pp.~145--146.

\bibitem{Davis2007}
{\sc Mark H.~A. Davis and David~G. Hobson}, {\em The range of traded option
  prices}, Math. Finance, 17 (2007), pp.~1--14.

\bibitem{Debnath1990}
{\sc Lokenath Debnath and Piotr Mikusi\'nski}, {\em Introduction to {H}ilbert
  spaces with applications}, Academic Press, 1990.

\bibitem{Engl1996}
{\sc Heinz~W. Engl, Martin Hanke, and Andreas Neubauer}, {\em Regularization of
  inverse problems}, Kluwer Academic Publishers, 1996.

\bibitem{Figlewski2008}
{\sc Stephen Figlewski}, {\em Estimating the implied risk neutral density for
  the {U.S.} market portfolio}, in Volatility and time series econometrics:
  essays in honor of {R}obert {F.} {E}ngle, Tim Bollerslev, Jeffrey~R. Russel,
  and Mark Watson, eds., Oxford University Press, 2008.

\bibitem{Halmos1963}
{\sc Paul~R. Halmos}, {\em What does the spectral theorem say?}, Am. Math.
  Mon., 70 (1963), pp.~241--247.

\bibitem{Hentschel2003}
{\sc Ludger Hentschel}, {\em Errors in implied volatility estimation}, J.
  Finan. Quant. Anal., 38 (2003), pp.~779--810.

\bibitem{Jackwerth2004}
{\sc Jens~Carsten Jackwerth}, {\em Option-implied risk neutral distributions
  and risk aversion}, Research Foundation of AIMR (CFA Institute), 2004.

\bibitem{Jackwerth1996}
{\sc Jens~Carsten Jackwerth and Mark Rubinstein}, {\em Recovering probability
  distributions from option prices}, J. Finance, 51 (1996), pp.~1611--1631.

\bibitem{Jarrow1982}
{\sc Robert Jarrow and Andrew Rudd}, {\em Approximate option valuation for
  arbitrary stochastic processes}, J. Finan. Econ., 10 (1982), pp.~347--369.

\bibitem{Merton1973}
{\sc Robert~C. Merton}, {\em Theory of rational option pricing}, Bell J. Econ.
  Manag. Sci., 4 (1973), pp.~141--183.

\bibitem{Musiela2008}
{\sc Marek Musiela and Marek Rutkowski}, {\em Martingale methods in financial
  modeling, $2^{nd}$ ed.}, Springer-Verlag, 2008.

\bibitem{Potters1998}
{\sc Marc Potters, Rama Cont, and Jean-Philippe Bouchaud}, {\em Financial
  markets as adaptive systems}, Europhys. Lett., 41 (1998), pp.~239--244.

\bibitem{Stutzer1996}
{\sc Michael Stutzer}, {\em A simple approach to derivative security
  valuation}, J. Finance, 51 (1996), pp.~1633--1652.

\end{thebibliography}

\end{document}